\newif\ifextended
\newif\ifanonymous
\newif\ifpoly
\newif\ifreview
\newif\ifcomments

\extendedtrue
\anonymousfalse
\polyfalse
\reviewfalse
\commentsfalse

\documentclass[acmsmall, \ifextended nonacm \else \fi, screen, prologue, divpsnames, svgnames,
  \ifreview review=true\else review=false\fi,
\ifanonymous anonymous=true\else anonymous=false\fi]{acmart}
\usepackage{graphicx} % Required for inserting images
\usepackage{enumerate}
\usepackage{proof}
\usepackage{mathtools}
\usepackage{amsmath}
\usepackage{amsthm}
\usepackage{multicol}
\usepackage{tikz}
\usepackage{natbib}
\usepackage{stmaryrd}
\usepackage{wrapfig}
\usepackage{syntax}
\usepackage{MnSymbol}
\usepackage{verbatim}
\usepackage{listings}
\usepackage{style}
\usepackage{tabularx}
\usepackage{appendix}
\usepackage[normalem]{ulem}
\usepackage{hyperref}
\hypersetup{
  colorlinks   = true,
  citecolor    = red
}

\definecolor{burntorange}{rgb}{0.8, 0.33, 0.0}
\definecolor{darkteal}{rgb}{0.0, 0.45, 0.45}

\theoremstyle{definition}
\newtheorem{theorem}{Theorem}[section]

\newtheorem{lemma}[theorem]{Lemma}
\newtheorem{definition}[theorem]{Definition}

\newcommand{\defas}{\mathrel{::=}}
\newcommand{\alt}{\:|\:}

\newcommand{\coeffcolor}[1]{\textcolor{burntorange}{#1}}

\newcommand{\downshift}[1]{\coeffcolor{\downarrow_{#1}}}

\newcommand{\utype}[1]{\textbf{U}_{\coeffcolor{#1}}}
\newcommand{\ftype}[0]{\textbf{F}}
\newcommand{\dom}[1]{\textit{dom } #1}

\newcommand{\qstrict}{\coeffcolor{\mathrm{S}}}
\newcommand{\qstrictarticle}{an }
\newcommand{\qlazy}{\coeffcolor{\mathrm{L}}}
\newcommand{\qunknown}{\coeffcolor{\mathrm{?}}}
\newcommand{\coeffoverline}[1]{\coeffcolor{\overline{#1}}}
\newcommand{\lazyvec}{\coeffoverline{\qlazy}}
\newcommand{\coeffsymbol}{\coeffcolor{\alpha}}

\newcommand{\cbpvlangname}{$\text{CBPV}^{\coeffcolor{\gamma\!}}$\xspace}
\newcommand{\cbnlangname}{$\text{CBN}^{\coeffcolor{\gamma\!}}$\xspace}

\ifextended \else
\setcopyright{cc}
\setcctype{by}
\acmDOI{10.1145/3776657}
\acmYear{2026}
\acmJournal{PACMPL}
\acmVolume{10}
\acmNumber{POPL}
\acmArticle{15}
\acmMonth{1}
\received{2025-07-08}
\received[accepted]{2025-11-06} \fi

\begin{document}

\ifextended \else 

 \fi

\ifextended \title{Typing Strictness (Extended Version)} \else \title{Typing Strictness} \fi
\date{}
\ifextended \author{\href{https://sainati.pl/}{Daniel Sainati}} \else \author{Daniel Sainati} \fi
\ifextended \else \orcid{0009-0003-4738-8524} \fi
\affiliation{%
  \institution{University of Pennsylvania}
  \city{Philadelphia}
  \country{USA}
}
\email{sainati@seas.upenn.edu}
\ifextended \author{\href{https://www.cis.upenn.edu/\%7e\%6a\%77\%63/}{Joseph W. Cutler}} \else \author{Joseph W. Cutler}  \fi
\ifextended \else \orcid{0000-0001-9399-9308} \fi
\affiliation{%
  \institution{University of Pennsylvania}
  \city{Philadelphia}
  \country{USA}
}
\email{jwc@seas.upenn.edu}
\ifextended \author{\href{https://www.cis.upenn.edu/\%7e\%62\%63\%70\%69\%65\%72\%63\%65/}{Benjamin C. Pierce}} \else \author{Benjamin C. Pierce} \fi
\ifextended \else \orcid{0000-0001-7839-1636} \fi
\affiliation{%
  \institution{University of Pennsylvania}
  \city{Philadelphia}
  \country{USA}
}
\email{bcpierce@seas.upenn.edu}
\ifextended  \author{\href{https://www.cis.upenn.edu/\%7e\%73\%77\%65\%69\%72\%69\%63\%68/}{Stephanie Weirich}} \else \author{Stephanie Weirich} \fi
\ifextended \else \orcid{0000-0002-6756-9168} \fi
\affiliation{%
  \institution{University of Pennsylvania}
  \city{Philadelphia}
  \country{USA}
}
\email{sweirich@seas.upenn.edu}

\begin{abstract}
  Strictness analysis is critical to efficient implementation of languages with non-strict evaluation,
  mitigating much of the performance overhead of laziness.
  However, reasoning about strictness at the source level can be challenging and unintuitive.
  We propose a new definition of strictness that
  refines the traditional one by describing variable usage more precisely.
  We lay type-theoretic foundations for this definition in both
  call-by-name and call-by-push-value settings,
  drawing inspiration from the literature on type systems tracking effects and coeffects.
  We prove via a logical relation that the strictness attributes computed by our type systems accurately
  describe the use of variables at runtime,
  and we offer a strictness-annotation-preserving translation from the call-by-name system to the call-by-push-value one.
  All our results are mechanized in Rocq.
\end{abstract}

\maketitle

\ifextended \section*{This paper is an extended version of \citet{typing-strictness}.} \fi

\section{Introduction} \label{sec:intro}
Non-strict evaluation offers benefits over both fully lazy and fully strict
evaluation strategies by allowing
expressions to be evaluated at any point between when they are first encountered and when
their value is needed.
Unlike strict strategies, which \textit{must} evaluate expressions
as soon as they are encountered,
non-strict strategies may delay evaluation arbitrarily,
improving performance (e.g., by skipping unnecessary computation)
and making it easier to work with codata like infinite lists.
Unlike lazy strategies,
which \textit{must} wait until an expression is used before evaluating it,
non-strict strategies may evaluate earlier
to avoid creating thunks.

These benefits have led a number of languages to support non-strict evaluation,
either as the default strategy, like Haskell \cite{Haskell} or R \cite{rlang},
or optionally via special annotations, like OCaml \cite{ocaml-lazy}.
Non-strict evaluation is also critical for the performance of
streaming libraries like FS2 in Scala \cite{fs2}
and code patterns like iterators in Rust \cite{rust-iterator},
and it is a useful tool for automatic parallelization \cite{calderon-trilla-improving-2015}
and for opportunistic execution in scripting languages \cite{mell2025}.

Non-strict evaluation also comes with a number of caveats.
Overly lazy strategies suffer from drawbacks
such as increased memory usage \cite{aviral-2019}, unpredictable behavior \cite{r-melts-brains} and space complexity \cite{Haskell-foldl, nothunks,devries-visualizing},
and security vulnerabilities \cite{Vassena-concurrent-2017}.
To circumvent these issues,
compilers like GHC \cite{ghc} perform \textit{strictness analysis}
\cite{mycroft-theory-1980,wadler-projections-1987},
allowing them to reorder code to evaluate eagerly when they can prove that doing so will not change observable behavior.
In pure settings,
this analysis enables performance optimizations that can massively reorder evaluation,
but such improvements are not limited only to languages like Haskell;
strictness analysis can still unlock speedups in impure languages when
the compiler is able to ensure that it will not reorder any side effects \cite{mell2025}.

Strictness analysis, however, comes with its own set of wrinkles.
Allowing the compiler to reorder evaluation can result in unpredictable behavior,
where modifying the usage of values in one place can change the
performance \cite{neilmitchell-2013} and even correctness \cite{johann-2004} of code elsewhere.
Additionally, the inherent imprecision of strictness analysis can result in
expressions that seem as though they should evaluate strictly but do not,
requiring the programmer to tweak their code to coax the analyzer into
triggering the desired optimization \cite{Haskell-seq}.
Exacerbating these issues is the fact that
the analysis' model of strictness
is specialized to compilation
and awkward as a reasoning tool:
as we explain in Section \ref{sec:motivation},
this model is overly extensional and can
misrepresent how variables are used.

To clarify the traditional model's murkiness regarding variable usage,
we propose a new, static notion of strictness, dubbed \textit{intensional strictness},
that enables more direct reasoning about how programs use variables;
% that more directly describes how programs use their variables;
any given occurrence of a variable can be either strict, lazy, or statically unknown.
We present intensional strictness type-theoretically
via a \textit{type-and-effect} system \cite{lucassen-polymorphic-1988},
taking additional inspiration from the literature on \textit{coeffects} \cite{petricek-coeffects-2013}.
Effect and coeffect annotations extend type systems to describe how programs interact with the world;
the standard intuition is that effects describe what a program \textit{does} to the world,
while coeffects describe what a program \textit{requires} from the world.
The interplay between effects and coeffects is
an active research area \cite{dal-lago-2022, gaboardi-combining-2016,torczon-effects-2024,McDermottMycroft-2018},
and we adopt mechanisms from these works to characterize the way that strict usage
both imposes a demand on a program's environment
(by requiring that a value be successfully produced),
and modifies that environment
(by triggering the evaluation of suspended computations).

Our new definition of strictness, while syntactic in nature like any type system,
captures semantic notions of strict and lazy usage that should be
intuitively familiar to users of Haskell and similar languages.
Furthermore, it refines the extensional characterization of strictness
offered by conventional strictness analysis, describing variable usage
more precisely and decoupling the definition from questions of ``demand'' or observable behavior.

We formalize our new notion of strictness via \cbnlangname,
a call-by-name (CBN) calculus with strictness effects,
and \cbpvlangname, an extension of call-by-push-value (CBPV)
\cite{levy-call-by-push-value-1999} with similar annotations.
CBPV is a convenient foundation for the study of strictness because it specifies explicitly where computation occurs by
separating computations from values and including explicit
constructs that suspend and resume execution.
Furthermore, a canonical translation exists from CBN into CBPV \cite{levy-call-by-push-value-2006},
allowing us to use \cbpvlangname as a tool to
understand strictness in CBN languages.

Our primary motivation, throughout, is conceptual: our proposed definition is
intended as a tool for better understanding and reasoning about
strictness, not yet as a basis for practical implementation or
surface-level design of non-strictly evaluated languages.
All of our proofs \cite{artifact} are mechanized in Rocq \cite{rocq}.

Concretely, our contributions are:

\begin{itemize}
  \item We propose a new notion of \textit{intensional strictness} and
    argue that it refines the extensional definition used by traditional strictness analysis (Section \ref{sec:motivation}).
  \item We present \cbnlangname,
    a \textit{type-and-effect system} for a call-by-name language
    that embodies this new form of strictness (Section \ref{sec:cbn}).
  \item We present \cbpvlangname, a variant of CBPV
    that also embodies intensional strictness.
    We instrument a big-step operational semantics with strictness attributes and
    prove the soundness of \cbpvlangname typing with respect to this
    semantics, thus providing a sound semantics for
    \cbnlangname via a
    translation to \cbpvlangname that preserves strictness annotations (Section \ref{sec:details}).
    \cbpvlangname allows us to factor our proofs about \cbnlangname
    through a lower-level intermediary,
    simplifying the proofs themselves and potentially offering a
    setting for reasoning about strictness in other languages with
    translations into \cbpvlangname.
  \item We prove that the strictness attributes computed by \cbpvlangname truly reflect strict and lazy usage:
    a well-typed program \textit{can} be run in an environment
    without a binding for any lazily-used variable
    and \textit{cannot} be run in an environment
    lacking any strictly-used variable.
    These proofs show that intensional strictness,
    as modeled by \cbnlangname and \cbpvlangname,
    refines the original extensional definition (Section \ref{sec:bot}).
    The proofs make use of a pair of logical relations and
    involve a rather delicate treatment of variables and scoping.
  \item We enrich \cbpvlangname and \cbnlangname with \textit{unused variable tracking},
    an extension that is surprisingly \ifextended \newpage \else \fi simple in the former language and surprisingly complex in the latter.
    All the previously proven theorems hold for this extension (Section \ref{sec:absence}).
\end{itemize}
Section \ref{sec:related} discusses related work.
Section \ref{sec:future} concludes and outlines possible future work.

\section{What Is Intensional Strictness?} \label{sec:motivation}

In this section,
we offer a new, \textit{intensional} definition of strictness
and examine how it differs from the extensional version used by
strictness analysis.
We build an intuition for how intensional strictness works
and explain why it is natural to model it type-theoretically.

We describe each syntactic occurrence of a variable
as either a \emph{strict} or a \emph{lazy} usage of that variable.
A strict variable usage scrutinizes the value stored in that variable
(e.g., it pattern matches on a pair or applies a function),
while a lazy usage does not (e.g., it places the variable into a cons cell of a lazy list).
A helpful analogy compares variables to boxes:
a strict use opens the box to observe its contents,
while a lazy use passes the box along unopened.

We might try to lift this terminology to call-by-name
functions---``a strict function is one that scrutinizes its argument''---but
this na\"{i}ve definition is ambiguous about
how the return values of functions are used.
It is unclear from this definition what the strictness of
the identity function $\lambda x.x$ should be, for example:
it does not directly scrutinize its argument, but
scrutinizing the result of a call to the identity \emph{would} scrutinize
the argument as well.
Instead of considering functions in isolation,
we need a more precise definition that allows us to
reason about the contexts in which functions are called and the demands made on their results.

The strictness-analysis literature \cite{mycroft-theory-1980,wadler-projections-1987}
attempts to provide this precision by defining strict functions
as those that fail to evaluate whenever their arguments fail to evaluate
and lazy ones as those that (might) succeed even if their argument fails.
This definition posits a hypothetical $\bot$ term,
intuitively representing the result of an erroring or non-terminating computation,
and says that a function $f$ is strict if $f \bot=\bot$; that is,
a strict function preserves failure.
Because it is defined in terms of observable behavior, we refer to
this notion of strictness as \textit{extensional strictness}.

Strictness-analysis algorithms based on extensional strictness
can mitigate much of the performance overhead of laziness
\cite{kubiak-implementing-1992,sergey-theory-2014}.
Compilers perform strictness analysis in an optimization pass,
wherein they try to determine if it is safe to evaluate an expression eagerly
and potentially save space by not allocating a thunk.
By ``safe,'' we mean that eager evaluation would not change the
observable behavior of the program,
compared to evaluating that same expression lazily.
This use case shows why the strictness-analysis literature uses
the definition it does: if $f \bot=\bot$, then $f$'s argument
can be evaluated eagerly, since, if the argument fails to evaluate,
the call will too.

As an example of how strictness analysis can improve performance,
consider an implementation of $\texttt{sum}$ for lists using a tail-recursive fold in
a Haskell-like language:
\texttt{sum lst = \textbf{foldl} (+) 0 lst}.
A programmer familiar with tail recursion might
expect this function to use constant stack space,
but, when evaluated lazily, it uses space linear in the length of the input.

\begin{wrapfigure}{r}{0.50\textwidth}
  \centering
  \ifextended \vspace*{-0.75\baselineskip} \else  \vspace*{-1.1\baselineskip}  \fi
  \begin{lstlisting}
    foldl (+) 0 [1 .. n]                            - ->
    foldl (+) (0 + 1) [2 .. n]                   - ->
    foldl (+) ((0 + 1) + 2) [3 .. n]           - ->
    foldl (+) ((((0 + 1) + 2) + ...) + n) []
  \end{lstlisting}
  \ifextended \vspace*{-1\baselineskip} \else \vspace*{-2.2\baselineskip} \fi
  \caption{Lazy evaluation of list sum using foldl}
  \vspace*{-0.75\baselineskip}
  \label{fig:strict-example-fold}
\end{wrapfigure}

To see why, let's step through a lazy evaluation of \texttt{sum} in Figure \ref{fig:strict-example-fold}.
Evaluation begins by applying \texttt{+} to the accumulator and the first element of \texttt{[1 .. n]},
the input to the \texttt{\textbf{foldl}}.
However, rather than completing that application and yielding an integer value as the new accumulator,
evaluation instead yields a thunk referencing both the old accumulator and the list element,
resulting in an ever-growing chain of nested thunks.
This chain is only forced once evaluation reaches the end of the list
and the result of the call to \texttt{sum} is used;
in the meantime, the entire contents of the list will be materialized within it.

If we recognize that the whole contents of the list will eventually be scrutinized by the
\texttt{+} operator, we can instead eagerly evaluate the
application to a value and maintain constant space without changing observable behavior;
if computation of any element of the list fails, we could never have computed the list's sum,
regardless of when that computation occurred.
Compilers like GHC use strictness analysis to
improve performance by exploiting this realization to
reorder evaluation
\cite{Haskell-strictness}.

However, despite its usefulness to compilers,
extensional strictness is otherwise unsatisfying as a tool for human reasoning about programs.
In particular, it does not directly say anything about how values are used!
A function that always fails (i.e., always returns $\bot$), for example,
is considered strict even if it never mentions its input,
despite the fact that strictly evaluating a call to
such a function can in fact change program behavior.
This characterization is counterintuitive;
to better understand strict and lazy usage,
we would like a model of strictness that does not conflate a function's
observable behavior with the way that it uses its argument.

To capture this additional nuance,
we propose a definition of \textit{intensional strictness} that directly describes how variables are used
within the bodies of functions. We assume that the results of function calls are scrutinized;
if the result of a call is used lazily, then the call will not be evaluated at all.
Intensional strictness is presented informally here in terms of functions;
later we will formalize it and generalize to open terms via the type systems in
Sections \ref{sec:cbn} and \ref{sec:details}.
\begin{definition}[Intensional Strictness] Functions can exhibit one of three kinds of strictness:
  \label{def:fn-strict}
  \begin{enumerate}
    \item A \textit{strict function} is one that, on all possible execution paths through its body,
      scrutinizes its argument whenever its return value is scrutinized.
    \item  A \textit{lazy function} is one that
      does not scrutinize its argument, regardless of how its return value is used.
    \item An \textit{indeterminate function} may or may not scrutinize its argument;
      the argument's usage is either statically unknown or
      varies depending on which execution path is taken when the function is called.
  \end{enumerate}
\end{definition}

\begin{figure}
  \begin{lstlisting}
    -- x is strict, a is indeterminate                  -- z is strict                     -- u is lazy                      -- y is lazy
    f1 x a = if x then a + 1 else 2                     f2 z = z                             f3 u = loop                      f4 y = Just y
  \end{lstlisting}
  \ifextended \vspace*{-1.5\baselineskip} \else \vspace*{-2\baselineskip} \fi
  \caption{Simple examples of variable usage}
  \ifextended \vspace*{-0.5\baselineskip} \else \vspace*{-1\baselineskip} \fi
  \label{fig:strict-example-1}
\end{figure}

To get a feel for this definition,
let's consider the code snippets in Figure \ref{fig:strict-example-1}.
What is the strictness of each of the parameters of these four functions,
assuming that their results are used strictly?

Scrutiny of \texttt{f1}'s result means we must know its value, and thus
we must evaluate \texttt{x}
to choose the correct arm of the conditional. Hence, \texttt{f1} is strict in \texttt{x}.
Similarly, \texttt{f2} is strict in \texttt{z} because it returns its argument:
if the result of a call to \texttt{f2} is scrutinized, then its argument is too.
The extensional definition agrees with these labels:
\texttt{f1}'s argument \texttt{x} can be evaluated eagerly without changing the program's observable
behavior, since its value will always be needed by \texttt{f1}'s body.
Similarly, the identity function \texttt{f2} is extensionally strict:
given $\bot$ as its argument, its result will also be $\bot$.

The definitions disagree, however, on the strictness of \texttt{u} in \texttt{f3}.
The $\bot$ term used by the extensional definition represents the result of a
failing or non-terminating computation,
and, according to this model, \texttt{f3} always returns $\bot$, as it will never terminate.
Accordingly, the extensional definition of strictness would describe \texttt{f3} as strict in its argument:
if \texttt{f3 u} = $\bot$ for all \texttt{u}, then \texttt{f3} $\bot$ = $\bot$ trivially.
However, \texttt{f3} does not actually mention \texttt{u} anywhere in its body, let alone use it strictly,
so the intensional definition would describe it as lazy,
more accurately describing how the function interacts with its argument.
This difference can matter in practice:
consider an application of \texttt{f3} to an argument that throws an exception.
Eagerly evaluating that argument would change the behavior of the function call,
causing an error where the call would otherwise run forever.
The extensional definition does not capture this nuance, and
the additional precision afforded to us in this case by reasoning
syntactically exemplifies the intensional nature of our new definition.

Our new definition of strictness also provides an opportunity to be more precise about
what lazy usage means by splitting the ``lazy'' label into two separate characterizations.
We describe \texttt{f4} as lazy in \texttt{y},
since constructors in non-strictly-evaluated languages do not scrutinize their arguments,
and scrutiny of \texttt{f4}'s return value
only requires evaluating to the top-level
constructor (in Haskell and similar languages, strict usage requires evaluation to weak head normal form \cite{Haskell-whnf}).
The traditional definition would also describe \texttt{f4} as lazy, but
laziness in the traditional sense only means that \texttt{f4}'s argument
\emph{might} not be scrutinized and hence cannot be pre-evaluated.
Definition \ref{def:fn-strict}, by comparison, provides more clarity by
telling us that \texttt{f4}'s argument is \emph{definitely} not scrutinized,
and thus that a call to \texttt{f4} will still succeed
even if its argument fails to produce a value.

In \texttt{f1}, on the other hand,
the argument \texttt{a} is scrutinized by \texttt{+},
but only in the success branch of the conditional;
in the failure branch, it is not used at all.
Definition \ref{def:fn-strict} tells us that \texttt{f1} is not strict in
\texttt{a},
as \texttt{a} is not used strictly on all possible execution paths.
However, \texttt{f1} is not lazy in \texttt{a} either, since
\texttt{a} may be scrutinized if \texttt{f1}'s return value is.
Traditional strictness analysis would label
\texttt{f1} as lazy in \texttt{a},
since eagerly evaluating \texttt{a} might change the function's behavior.
But this is misleading; some calls to \texttt{f1} may actually scrutinize \texttt{a}.
It is more precise to say that \texttt{f1}'s strictness is \textit{indeterminate}
with respect to \texttt{a}, reflecting the fact that \texttt{a} may be
used differently by different calls to \texttt{f1}.
In general, the ``indeterminate'' classification used by intensional
strictness provides the same information as the traditional ``lazy'' label,
while the intensional ``lazy'' label provides additional precision.

These examples illustrate the ways in which intensional strictness
refines the extensional definition embodied in existing strictness-analysis algorithms:
it more precisely characterizes cases that the traditional definition would label lazy,
splitting the usual notion of lazy usage into ``indeterminate'' use and ``definitely lazy'' use.
Furthermore, it decouples reasoning about
strictness from orthogonal questions of extensional behavior,
describing variable use intensionally and distinguishing between the different ways that computation can fail.

Definition \ref{def:fn-strict} is informal, however:
we still need a rigorous and formal model to undergird it.
To generalize our reasoning about variables and functions to
arbitrary terms and positions, the model must be compositional and syntactic,
so a type-theoretic approach is naturally suited for this task.
However, it is not clear whether existing type systems like usage typing \cite{turner-1995, wansbrough-1999, wansbrough-2005, girard-linear, Petricek-coeffect-2014}
and information-flow typing \cite{denning-denning-info,volpano-secure-flow,Zdancewic-2001,zdancewic-secure-002,Palsberg-flow}
are capable of modeling intensional strictness (for reasons detailed in Section \ref{sec:related}),
so we will need to develop a type system of our own.

\vspace*{-0.5em}

\section{A Type System for Strictness}  \label{sec:cbn}

\newcommand{\cbnvdash}{\vdash_{\texttt{CBN}}}

\begin{figure}
  \[
    \begin{array}{llcl}
      \coeffcolor{\textit{strictness attributes }} & \coeffsymbol & \defas & \qstrict \alt \qlazy \alt \qunknown \\
      \coeffcolor{\textit{strictness effects }} & \coeffcolor{\gamma} & \defas & \coeffcolor{\cdot} \alt \coeffcolor{\gamma, x : \coeffsymbol} \\
      \textit{types } & \tau & \defas & \texttt{unit} \alt
      \tau_1^{\coeffcolor{\gamma_1}} \times \tau_2^{\coeffcolor{\gamma_2}} \alt
      \tau_1^{\coeffcolor{\gamma_1}} + \tau_2^{\coeffcolor{\gamma_2}} \alt
      (x :^{\coeffsymbol} \tau^{\coeffcolor{\gamma_1}}) \xrightarrow{\coeffcolor{\gamma_2}} \tau\\
      \textit{contexts } & \Gamma & \defas & \cdot \alt \Gamma, x : \tau^{\coeffcolor{\gamma}} \\
      & & & \\
      \textit{expressions } & e & \defas & () \alt x \alt \textbf{inl } e \alt
      \textbf{inr } e \alt (e, e) \alt \lambda x. e \alt e_1\;e_2 \alt \textbf{let } x = e \textbf{ in } e  \alt \textbf{sub } e \\
      & & & \alt e; e \alt \textbf{let } (x_1, x_2) = e \textbf{ in } e
      \alt \textbf{case } e \textbf{ of inl } x_1 \rightarrow e, \textbf{ inr } x_2 \rightarrow e
    \end{array}
  \]
  \ifextended  \vspace*{-0.5\baselineskip} \else \vspace*{-0.5\baselineskip} \fi
  \caption{Types and Syntax of \cbnlangname}
  \ifextended \vspace*{-1.5\baselineskip} \else \vspace*{-0.5\baselineskip} \fi
  \label{fig:cbn-syntax}
\end{figure}

We present a call-by-name language, \cbnlangname,
that tracks intensional strictness in its type system using effects.
The types and syntax are given in Figure \ref{fig:cbn-syntax}.
As a visual aid,
we typeset the parts of \cbnlangname that are specific to strictness in \coeffcolor{orange},
since these are the novel aspects of the system.
Parts in black are standard.
Along with the usual type introduction and elimination forms,
the language includes a sequencing operation $e_1;\;e_2$, which strictly uses the \texttt{unit}-typed
result of $e_1$ before continuing with $e_2$.

\begin{wrapfigure}{r}{0.30\textwidth}
  \centering

  \vspace*{-1\baselineskip}

  \begin{tikzpicture}

    \node at (0,0)    (U)  {$\qunknown$};
    \node at (0.75,0.75)    (S)  {$\qstrict$};
    \node at (-0.75,0.75)   (L)  {$\qlazy$};

    \draw (U)   -- (L);
    \draw (U)   -- (S);
  \end{tikzpicture}

  \vspace*{-0.75\baselineskip}

  \caption{Strictness semilattice}
  \label{fig:lattice}
  \vspace*{-0.75\baselineskip}
\end{wrapfigure}

\cbnlangname assigns every variable one of three \textit{strictness attributes} to track its usage,
ordered according to the semilattice depicted in Figure
\ref{fig:lattice}.  The metavariable $\coeffsymbol$ ranges over
strictness attributes.
The $\qstrict$ attribute asserts that a variable is definitely used
strictly at least once on every execution path.
The $\qlazy$ attribute asserts the opposite: that a variable is used
only lazily (or not at all) on every execution path.\footnote{
  To simplify the type system, we assign both lazily-used and unused variables the $\qlazy$ attribute. Later in Section \ref{sec:absence}
  we will show how to distinguish these two kinds of use at the cost of some additional complexity.
}
The $\qunknown$ attribute asserts nothing; a variable with this annotation
may be used strictly, or lazily, or sometimes one and sometimes the
other (on different execution paths), or not at all.
We need the last of these because static strictness tracking
(via types or otherwise) is inherently imprecise;
in the presence of branching it is not always possible to be certain
of how a variable will be used.

However, these attributes alone are not sufficient to track strictness.
As \citet{McDermottMycroft-2018} explain,
scrutinizing variables in a non-strictly evaluated language can cause
other code to be evaluated,
and we must track how this evaluation uses variables as well.
Accordingly, besides the attributes that appear on \cbnlangname types and variables
to describe their usage,
types also include \textit{attribute vectors} (denoted with the metavariable $\coeffcolor{\gamma}$)
that describe how all in-scope variables are used by the
suspended computations that inhabit a given type.
These vectors $\coeffcolor{\gamma}$ are \textit{effects}
modeling the additional evaluation triggered by variable scrutiny;
the type system tracks these effects by associating types and variables
with $\coeffcolor{\gamma}$s describing what will happen when their values are scrutinized.
(Further discussion of this terminology can be found in Section \ref{subsec:naming}.)
We write the vector that maps all in-scope variables to $\qlazy$ as $\lazyvec$, and,
to reduce clutter, omit $\coeffcolor{\gamma}$s in examples when they are $\lazyvec$.
Additionally, we freely omit any variables with $\qlazy$ attributes
from $\coeffcolor{\gamma}$s. That is,
$\coeffcolor{x:\qstrict, y :\qlazy}$ and $\coeffcolor{x:\qstrict}$ denote the same vector.

\subsection{A Few Examples} \label{subsec:examples} Before exploring the typing rules of \cbnlangname in detail,
let's look at a few examples to build intuition.
For each example term $e$, we'll describe how \cbnlangname types the term using a
judgment of the form $\Gamma \cbnvdash e :^{\coeffcolor{\gamma}} \tau$.
The $\Gamma$ in this judgment is a typing context,
$\tau$ is the type of the term,
and $\coeffcolor{\gamma}$ maps each
free variable in $e$ to a strictness attribute describing how it is used.

\subsubsection*{Addition}

As an introductory example, consider the term $x + 1$.
This term uses $x$ strictly, as addition requires the value of $x$ to be known in order to add to it.
\cbnlangname checks this term using the judgment
$$x:\texttt{Int}  \cbnvdash x + 1 :^{\coeffcolor{x:\qstrict}} \texttt{Int}.$$
The orange $\coeffcolor{x:\qstrict}$ that appears on the $:$ is where we find the strictness
information for the term being checked; $x$ is assigned an $\qstrict$
attribute, reflecting its strict usage.

\subsubsection*{Pairs and $\eta$-Laws}

Next, we consider more complex data such as pairs.
Because constructors in \cbnlangname
do not require their arguments to be values, the types for such data
must contain information about how their components use their free variables.
Thus, \cbnlangname's types store separate $\coeffcolor{\gamma}$s for each of their components
describing how they would use any in-scope variables if evaluated.

Consider a pair $(z, \texttt{true})$ in a context where $z$ has the type $\texttt{Bool}$.
This term itself uses $z$ lazily (i.e., with effect $\coeffcolor{z:\qlazy}$),
but we also need to know how the components of this term would use $z$,
were we to scrutinize them.
\cbnlangname assigns this term the type
$\texttt{Bool}^{\coeffcolor{z:\qstrict}} \times \texttt{Bool}$---let's
call it $P$ for short.

$P$ describes pairs where
each side contains a thunk that, when evaluated,
will produce a boolean.
The orange vectors in $P$
tell us that the two sides of the product use the in-scope
variable $z$ differently. If $P$'s first element is used strictly,
$z$ will be used strictly, but if its second element is used strictly,
$z$ will be used only lazily.
\cbnlangname types this term with the judgment
$$z:\texttt{Bool} \cbnvdash (z, \texttt{true}) :
\texttt{Bool}^{\coeffcolor{z:\qstrict}} \times
\texttt{Bool}.$$
In general, the attributes on a judgment's :
describe how a term uses its free variables \textit{now}, while
those in the term's type describe how
it might use its free variables \textit{later} as more of it is scrutinized.

Pairs and other structured types in \cbnlangname exhibit a crucial difference from other
type systems for CBN calculi.
Consider the terms $x$ and $( \textbf{fst } x, \textbf{snd } x)$
in a context where $x$ has type $\texttt{Bool} \times \texttt{Bool}$.
In a pure CBN calculus, these two terms are considered
$\eta$-equivalent and have the same type.
In effectful settings like this one, however, the $\eta$-equivalence of these terms breaks down;
any effects that would result from evaluating $x$
occur only when scrutinizing the first term.
The pair constructor in the second term suspends the uses of $x$,
so $x$ would not be evaluated at all if this term's top-level constructor were scrutinized.
In short, the first term uses $x$ strictly, while the second uses it lazily.
\cbnlangname reflects this by requiring the two terms to be typed
differently: $x$ is typed with
$$x:\texttt{Bool} \times \texttt{Bool}
\cbnvdash x :^{\coeffcolor{x:\qstrict}} \texttt{Bool} \times \texttt{Bool},$$
whereas $( \textbf{fst } x, \textbf{snd } x)$ is typed with
$$x:\texttt{Bool} \times \texttt{Bool}
\cbnvdash ( \textbf{fst } x, \textbf{snd } x) :
\texttt{Bool}^{\coeffcolor{x:\qstrict}} \times \texttt{Bool}^{\coeffcolor{x:\qstrict}}.$$
While the first term uses $x$ strictly,
the second term uses $x$ lazily and actually has a different type, one that
also says that it uses $x$ strictly in each of its components.

\subsubsection*{Contexts, Functions, and Latent Effects}

So far, we have not considered the fact that, in non-strictly evaluated languages, accessing
a variable may cause a thunk to be evaluated and thus other variables
to be scrutinized.
To track this deferred usage, typing contexts $\Gamma$ map variables to both types $\tau$
and vectors $\coeffcolor{\gamma}$; these $\coeffcolor{\gamma}$s describe
the strictness with which any in-scope variables will be used if the thunked expression
stored in a given variable is strictly used.
Consider a program
$\textbf{let } x = \textbf{if } y \textbf{ then } \texttt{true} \textbf{ else }\texttt{false} \textbf{ in } x$.
Assuming $y$ is in scope with type $\texttt{Bool}$,
$x$ is also a \texttt{Bool}, but one that uses $y$ strictly if evaluated.
Zooming in on the \textbf{let}'s body, \cbnlangname checks
$$y: \texttt{Bool}, x: \texttt{Bool}^{\coeffcolor{y:\qstrict}} \cbnvdash
x :^{\coeffcolor{y:\qstrict, x:\qstrict}} \texttt{Bool}.$$
In the context, $x$ is associated both with its type and with a
vector $\coeffcolor{y:\qstrict}$, capturing the fact that any strict use of $x$
will cause a strict use of $y$.
This fact is also reflected in the judgment's output effect,
which tells us that strict use of the term being checked will cause
a strict use of $y$.

Extending this principle from open terms to functions, it is apparent that
function types in \cbnlangname must also describe how their
arguments use other variables.
Function types in \cbnlangname thus
have the shape $(x :^{\coeffsymbol} \tau_1^{\coeffcolor{\gamma_1}}) \xrightarrow{\coeffcolor{\gamma_2}} \tau_2$.
The $\coeffsymbol$ (either $\qstrict$, $\qlazy$ or $\qunknown$) above the $:$ describes how the function uses its argument $x$,
while the vector $\coeffcolor{\gamma_2}$ describes the ``latent effect''
that will be produced when the function is called.
The vector $\coeffcolor{\gamma_1}$, on the other hand,
describes the effect that
the function's argument can have, i.e., how other variables will be used when the argument is evaluated.
Note also that this function type is dependent, in the sense that $\tau_2$ may mention $x$ in its attribute vectors.
A more detailed explanation of why this is necessary can be found in Section \ref{subsec:cbntyping}.

So, for example, one possible judgment for the identity function would be
$$y: \texttt{Bool} \cbnvdash
\lambda x.\;x :
(x:^{\qstrict} \texttt{Bool}^{\coeffcolor{y:\qstrict}}) \xrightarrow{\coeffcolor{y:\qstrict}} \texttt{Bool}.$$
This typing allows the identity to accept a boolean argument that uses $y$ strictly when evaluated,
and it tells us that the function uses its argument strictly and will use $y$ strictly when called.
The judgment also tells us that the function definition itself uses $y$ lazily,
since the use of the function's argument in the body is not evaluated if the function is not
called.\footnote{Note that this typing also would mean that this identity function could not accept
  an argument that used $y$ lazily, as the function's type would then be inaccurate
  about how its body uses $y$.
  Requiring function types to declare the
  effects associated with their arguments is restrictive,
  but prior work \cite{McDermottMycroft-2018}
  has addressed this limitation via effect polymorphism \cite{rytz-lightweight-2012,lucassen-polymorphic-1988},
  and, as we discuss in Section \ref{sec:future},
  the same approach should work here.
  Polymorphism would be an orthogonal addition to \cbnlangname,
  however, and
  we focus on monomorphic type systems here for clarity and
simplicity.}

\subsubsection*{Unpacking a Pair}
Consider the term $\textbf{let } (a, b) = x \textbf{ in } b$
in a context where $z$ has type $\texttt{Bool}$
and $x$ has type $P$ (that is,
  $\texttt{Bool}^{\coeffcolor{z:\qstrict}} \times
\texttt{Bool}$).
% \bcp{Might be easier just to
% write out the context?}
% \dhs{referring to P here lets me reuse it again later which saves a
% bunch of space}
% BCP: OK to leave it as is.  I understand more of the constraints now.
Assume also that $x$ does not use $z$ strictly when evaluated.
The following term unpacks a pair and returns its second element, and
\cbnlangname types it with
$$ z:\texttt{Bool}, x:P \cbnvdash
\textbf{let } (a, b) = x \textbf{ in } b :^{\coeffcolor{x:\qstrict}} \texttt{Bool}.$$
This judgment tells us that the term uses $x$ strictly
(since it is destructured by let-binding) but uses $z$ lazily.
This latter fact is a result of $x$'s type $P$,
which asserts that
strict use of its second element results in lazy usage of $z$.
When the term uses $b$ (the second element of $x$) strictly,
that usage therefore results in a lazy usage of $z$.
Note that, as the variable $b$ itself is local to this term,
we do not need to report $b$'s usage in the final typing
judgment.

Now consider, instead, a term that returns $x$'s first element:
$\textbf{let } (a, b) = x \textbf{ in } a$.
We could not type this term with the same judgment we previously used.
$P$ tells us that strict use of $a$ causes a strict use of $z$,
and therefore this term would have to be typed with the judgment
$$ z:\texttt{Bool}, x:P \cbnvdash
\textbf{let } (a, b) = x \textbf{ in } a :^{\coeffcolor{z:\qstrict, x:\qstrict}} \texttt{Bool}.$$

\subsubsection*{Indeterminate use}
Finally, consider a term that returns either the first or second element of
a $P$-typed variable $x$, depending on the value of another variable $y$:
$$\textbf{let } (a, b) = x \textbf{ in } \textbf{if } y \textbf{ then } a \textbf{ else } b.$$
What type can we give to this term?
Clearly it uses both $y$ and $x$ strictly,
but the way it uses $z$ depends on which branch of the $\textbf{if}$ statement is taken.
The success branch uses $z$ strictly, while the failure branch uses it lazily,
so we say that this term is \textit{indeterminate} in its usage of $z$.
As written, we cannot assign it a type,
but we can make a small modification and produce another term like so:
$$\textbf{let } (a, b) = x \textbf{ in }
\textbf{if } y \textbf{ then } (\textbf{sub }a) \textbf{ else } (\textbf{sub }b).$$
The \textbf{sub} annotation makes explicit that we are
using \textit{subsumption} when typing this term;
we forget information about how $z$ is used in each branch
(moving down the semilattice in Figure \ref{fig:lattice})
to derive
$$
z:\texttt{Bool}, y:\texttt{Bool}, x:P \cbnvdash
\textbf{let } (a, b) = x \textbf{ in }
\textbf{if } y \textbf{ then } (\textbf{sub }a) \textbf{ else } (\textbf{sub }b)
:^{\coeffcolor{y:\qstrict,x:\qstrict, z:\qunknown}} \texttt{Bool}.
$$
This judgment asserts that the term uses $y$ and $x$ strictly,
but says that the usage of $z$ is not known to be either strict or lazy.
The explicit \textbf{sub} syntax is useful for the proof of Lemma \ref{lem:fundamental-lemma-strict},
which requires inversion of the semantic judgment; without it the language would lose its
syntax-directed character and would require additional work to make its semantic rules invertible.
However, this syntax is included purely for technical convenience and
does not change the system's expressiveness.

%Note that the $\qunknown$ attribute describes variable usages that
%differ across potential execution paths; another possible intuition for this
%attribute is ``either strict or lazy'', rather than ``unknown''.
%In this framing, the counterpart to $\qunknown$ that represents both strict and lazy usage
%within the same execution path is just $\qstrict$: using a variable both
%lazily and strictly imposes the same requirements on its value
%that just a strict usage does.

\subsection{\cbnlangname Overview} \label{subsec:cbntyping}

\newcommand{\cbntypinginlinr}{
  $$
  \infer[\rlabel*{T-CBN-Inl}]
  {\Gamma \cbnvdash \textbf{inl } e :^{\lazyvec} \tau_1^{\coeffcolor{\gamma_1}} + \tau_2^{\coeffcolor{\gamma_2}}}
  {\Gamma \cbnvdash e :^{\coeffcolor{\gamma_1}} \tau_1 } \qquad
  \infer[\rlabel*{T-CBN-Inr}]
  {\Gamma \cbnvdash \textbf{inr } e :^{\lazyvec} \tau_1^{\coeffcolor{\gamma_1}} + \tau_2^{\coeffcolor{\gamma_2}}}
  {\Gamma \cbnvdash e :^{\coeffcolor{\gamma_2}} \tau_2 }
  \vspace{2.75ex}
  $$
}

\newcommand{\cbntypingpairseqsplit}{
  $$
  \infer[\rlabel*{T-CBN-Unit}]
  {\Gamma \cbnvdash () :^{\lazyvec} \texttt{unit}}
  {} \qquad
  \infer[\rlabel*{T-CBN-Seq}]
  {\Gamma \cbnvdash e_1 ; e_2 :^{\coeffcolor{\gamma_1 + \gamma_2}} : \tau}
  {\Gamma \cbnvdash e_1 :^{\coeffcolor{\gamma_1}} \texttt{unit} &
  \Gamma \cbnvdash e_2 :^{\coeffcolor{\gamma_2}} \tau}
  \vspace{2.75ex}
  $$
  $$
  \infer[\rlabel*{T-CBN-Split}]
  {\Gamma \cbnvdash \textbf{let } (x_1, x_2) = e_1 \textbf{ in } e_2
  :^{\coeffcolor{\gamma_1 + (\downshift{x_1} \downshift{x_2} \gamma_2)}} \downshift{x_1} \downshift{x_2} \tau}
  {\Gamma \cbnvdash e_1 :^{\coeffcolor{\gamma_1}} \tau_1^{\coeffcolor{\gamma_1'}} \times \tau_2^{\coeffcolor{\gamma_2'}} &
    \Gamma, x_1 : \tau_1^{\coeffcolor{\gamma_1'}}, x_2 : \tau_2^{\coeffcolor{\gamma_2'}}
  \cbnvdash e_2 :^{\coeffcolor{\gamma_2}} \tau}
  \vspace{2.75ex}
  $$
}

Having worked through some examples to build intuition,
we can begin to discuss the details of how \cbnlangname derives typing judgments.
Strictness attributes are elements of a \textit{preordered monoid}
equipped with a $\coeffcolor{+}$ operator to combine attributes and a $\coeffcolor{\leq}$ comparator to order them.
The $\coeffcolor{+}$ operator has $\qlazy$ as its identity and is defined in Table \ref{table:plus}.

\begin{wraptable}{r}{0.3\textwidth}
  \ifextended \else \vspace*{-1.75\baselineskip} \fi
  \ifextended \else \caption{Definition of $\coeffcolor{+}$}\label{table:plus} \fi 
  \ifextended \else \vspace*{-1\baselineskip} \fi
  \centering
  \begin{tabular}{l|llll}
    $\coeffcolor{+}$ & $\qstrict$ & $\qunknown$ & $\qlazy$ \\ \cline{1-4}
    $\qstrict$ & $\qstrict$ & $\qstrict$ & $\qstrict$  \\
    $\qunknown$ & $\qstrict$ & $\qunknown$ & $\qunknown$  \\
    $\qlazy$ & $\qstrict$ & $\qunknown$ & $\qlazy$
  \end{tabular}

  \captionsetup{justification=centering}
  \ifextended \caption{Definition of $\coeffcolor{+}$}\label{table:plus} \fi 
  \ifextended \vspace*{-2\baselineskip} \else \vspace*{-0.75\baselineskip} \fi
\end{wraptable}

The $\coeffcolor{\leq}$ comparator orders attributes according to the semilattice depicted in Figure \ref{fig:lattice},
which is notably different from the order induced by $\coeffcolor{+}$.
Instead, $\coeffcolor{\leq}$ orders attributes by information: one attribute is greater than
another when it gives us more certainty about how a variable will be used.

We lift $\coeffcolor{+}$ and $\coeffcolor{\leq}$ to operate over vectors $\coeffcolor{\gamma}$ pointwise,
noting that $\coeffcolor{\leq}$ forms a preorder over $\coeffcolor{\gamma}$s,
where a vector $\coeffcolor{\gamma_1}$ is only considered $\coeffcolor{\leq}$ a vector $\coeffcolor{\gamma_2}$
if every variable in $\coeffcolor{\gamma_1}$ has an attribute that is $\coeffcolor{\leq}$
the corresponding attribute in $\coeffcolor{\gamma_2}$.
So, for example, the vectors $\coeffcolor{x:\qunknown, y :\qlazy}$ and $\coeffcolor{x:\qlazy, y :\qunknown}$
are unrelated.\footnote{This still endows $\coeffcolor{\gamma}s$ with sufficient structure to be elements of a preordered monoid,
a requirement for them to behave like effects \cite{katsumata-2014}.}

\begin{figure}
  $$
  \infer[\rlabel*{T-CBN-Var}]
  {\Gamma \cbnvdash x :^{\coeffcolor{\gamma + x:\qstrict}} \tau}
  {x : \tau^{\coeffcolor{\gamma}} \in \Gamma} \qquad
  \infer[\rlabel*{T-CBN-Pair}]
  {\Gamma \cbnvdash (e_1, e_2) :^{\lazyvec}
  \tau_1^{\coeffcolor{\gamma_1}} \times \tau_2^{\coeffcolor{\gamma_2}}}
  {\Gamma \cbnvdash e_1 :^{\coeffcolor{\gamma_1}} \tau_1 &
  \Gamma \cbnvdash e_2 :^{\coeffcolor{\gamma_2}} \tau_2}
  \vspace{2.75ex}
  $$
  \ifextended
  \cbntypinginlinr
  \fi
  $$
  \infer[\rlabel*{T-CBN-Sub}]
  {\Gamma \cbnvdash  \textbf{sub } e :^{\coeffcolor{\gamma}} \tau}
  {\Gamma \cbnvdash e :^{\coeffcolor{\gamma'}} \tau & \coeffcolor{\gamma \leq \gamma'}} \qquad
  \infer[\rlabel*{T-CBN-Let}]
  {\Gamma \cbnvdash \textbf{let } x = e_1 \textbf{ in } e_2 :^{\coeffcolor{(\downshift{x}\gamma_2)}} \downshift{x} \tau_2}
  {\Gamma \cbnvdash e_1 :^{\coeffcolor{\gamma_1}} \tau_1 &
  \Gamma, x : \tau_1^{\coeffcolor{\gamma_1}} \cbnvdash e_2 :^{\coeffcolor{\gamma_2}} \tau_2}
  \vspace{2.75ex}
  $$
  $$
  \infer[\rlabel*{T-CBN-Abs}]
  {\Gamma \cbnvdash \lambda x. e :^{\lazyvec}
  (x :^{\coeffsymbol} \tau_1^{\coeffcolor{\gamma_1}}) \xrightarrow{\coeffcolor{\gamma_2}} \tau_2}
  {\Gamma, x :^{\coeffcolor{\gamma_1}} \tau_1 \cbnvdash e :^{\coeffcolor{\gamma_2, x : \coeffsymbol}} \tau_2 } \qquad
  \infer[\rlabel*{T-CBN-App}]
  {\Gamma \cbnvdash e_1\;e_2 :^{\coeffcolor{\gamma_2 + \gamma_3}}  \downshift{x} \tau_2}
  {\deduce{\Gamma \cbnvdash e_1 :^{\coeffcolor{\gamma_3}}
    (x :^{\coeffsymbol} \tau_1^{\coeffcolor{\gamma_1}}) \xrightarrow{\coeffcolor{\gamma_2}} \tau_2}
    {\Gamma \cbnvdash e_2 :^{\coeffcolor{\gamma_1}} \tau_1}
  }
  \vspace{2.75ex}
  $$
  \ifextended
  \cbntypingpairseqsplit
  \fi
  $$
  \infer[\rlabel*{T-CBN-Case}]
  {\Gamma \cbnvdash \textbf{case } e_1 \textbf{ of inl } x_1 \rightarrow e_2,
  \textbf{ inr } x_2 \rightarrow e_3 :^{\coeffcolor{\gamma_1 + \gamma_2}} \tau}
  {\deduce
    {\deduce
      {}
      {\Gamma, x_1: \tau_1^{\coeffcolor{\gamma_1'}}, \cbnvdash e_2 :^{\coeffcolor{\gamma_2, x_1:\coeffsymbol_1}} \tau_1' \;\;&\;\;
        \Gamma, x_2: \tau_2^{\coeffcolor{\gamma_2'}}, \cbnvdash e_3 :^{\coeffcolor{\gamma_2, x_2:\coeffsymbol_2} } \tau_2'
      }
    }
    {\Gamma \cbnvdash e_1 :^{\coeffcolor{\gamma_1}} \tau_1^{\coeffcolor{\gamma_1'}} + \tau_2^{\coeffcolor{\gamma_2'}} \;\;&\;\;
    \coeffcolor{\tau = \downshift{x_1} \tau_1' = \downshift{x_2} \tau_2'} }
  }
  $$
  \caption{Typing rules for \cbnlangname}
  \ifextended \vspace*{-0.75\baselineskip} \else \fi
  \label{fig:cbn-typing}
\end{figure}

\cbnlangname's type system incorporates effects $\coeffcolor{\gamma}$ into its typing judgment to track
the strictness of each variable in the context;
\ifextended \else a selection of \fi the typing rules can be found in Figure \ref{fig:cbn-typing}.
We give a named presentation of these typing rules here,
but our Rocq mechanization \cite{artifact}
uses a de Bruijn representation.

The simplest rule of interest is \rref{T-CBN-Var},
which looks up $x$ in the context $\Gamma$ and produces the effect $\coeffcolor{\gamma}$ that was latent there.
In addition to this effect, the rule also adds a strict usage of $x$ to its output effect
to reflect that the value stored in $x$ was strictly used.

Next, since let-binding is lazy,
the \rref{T-CBN-Let} rule does not realize the $\coeffcolor{\gamma_1}$ effects of the let-bound expression $e_1$.
Rather, it associates those effects with the bound variable $x$ in the context,
to be produced if that variable is strictly used later.
The rule realizes only the effects $\coeffcolor{\gamma_2}$ associated with $e_2$,
as any strict uses of $x$ in $e_2$ will already have produced $\coeffcolor{\gamma_1}$,
which will be included in $\coeffcolor{\gamma_2}$.

Also of note is how \rref{T-CBN-Let} handles scoping in its result.
During the checking of $e_2$, we must have an effect in the context for $x$.
However, once $x$ goes out of scope,
its usage is no longer of interest, and we drop its attributes from $\coeffcolor{\gamma_2}$ and $\tau_2$
via the $\coeffcolor{\downshift{x}}$ operator.
This operator traverses the structure of a type and removes any mappings for the variable $x$ from
all vectors in the type.
To understand why this is reasonable, consider that our goal in modeling
Definition \ref{def:fn-strict} is to know the strictness of each function's
argument---that is, whether an argument is used strictly or lazily by a function's body.
Given this,
the strict or lazy usage of a function's local variable is not of
external interest
and can be discarded when that variable goes out of scope.
Similarly, once we leave the scope of a let expression, we no longer
care to track the strictness of the variable it binds.
A detailed example of this behavior can be found in Section \ref{subsec:typing}.

The function introduction and elimination rules, $\rref{T-CBN-Abs}$ and $\rref{T-CBN-App}$,
are where we achieve the primary desideratum of the system:
the ability to annotate arrow types with the strictness of the functions they classify.
In the \rref{T-CBN-Abs} rule,
the body of the function is checked in a context where the input effect
$\coeffcolor{\gamma_1}$ is associated with the parameter $x$.
This will realize effects for the body,
which we can split into the portion $\coeffcolor{\gamma_2}$ not mentioning $x$
plus whatever $\coeffsymbol$ is associated with $x$ directly,
which tells us how $x$ is used in the function body.

The function type used in this rule is dependent because \cbnlangname function types
need to mention the variables they close over. Consider
the function $\lambda x.\;\lambda y.\;x$, where $x$ has type $\tau_1$ and $y$ type $\tau_2$.
This function is lazy in its argument; when partially applied it will yield the function
$\lambda y. x$ and will not evaluate its argument.
However, the inner function is strict in its free variable $x$, and its type must reflect that.
It therefore must have type $(y:^{\qlazy} \tau_2) \xrightarrow{\coeffcolor{x:\qstrict}} \tau_1$,
which then requires us to type the overall function at $(x:^{\qlazy} \tau_1) \xrightarrow{}
(y:^{\qlazy} \tau_2) \xrightarrow{\coeffcolor{x:\qstrict}} \tau_1$.
We see here that the return type of the overall function
(that is, the type of the inner function) must be able to mention the variable that the
overall function abstracts in order to properly describe its usage.
On the other hand, $x$ does \textit{not} need to appear
on the $\xrightarrow{}$ portion of the outer function type,
as $x$'s lazy usage in the outer function body is already
reflected in the $\qlazy$ associated with the argument.

The \rref{T-CBN-App} rule
enforces that the effects realized by the argument expression $e_2$,
namely $\coeffcolor{\gamma_1}$, match the effects expected by the function.
These effects are already included in the latent effects $\coeffcolor{\gamma_2}$ of the function type,
however. Therefore, as in the \rref{T-CBN-Let} rule, they do not need to be added
to the resulting effects for this rule.
Instead, the rule adds together only the effects realized by the body of the function
and the effects realized by the computation of the function $e_1$.

The \rref{T-CBN-App} rule also manages its scope similarly to
the \rref{T-CBN-Let} rule by dropping the abstracted variable from the result of the
function's return type after the call.
The $\coeffsymbol$ on the applied function type does not appear in
the rule's conclusion for similar reasons:
the calling scope does not need to track how the function body uses its argument.

The last point of note concerns effect subsumption.
Some imprecision is unavoidable in \cbnlangname as a result of the branching
in the \textbf{case} expression. We could soundly restrict
this imprecision to the \rref{T-CBN-Case} rule, but we instead allow
subsumption to happen anywhere via the
\rref{T-CBN-Sub} rule,
as it allows us to type more programs.

Consider, for example, the standard Church encoding of the booleans,
$\lambda x. \lambda y. x$ and $\lambda x. \lambda y. y$ \cite{church-boolean}.
To be a \textit{usable} encoding of the booleans,
both of these functions must inhabit the same type,
but without subsumption we would be forced to assign them the
types
\ifextended \else \vspace*{-.25\baselineskip} \fi
$$\lambda x. \lambda y. x \in (x:^{\qlazy} A) \xrightarrow{}
(y:^{\qlazy} A) \xrightarrow{\coeffcolor{x:\qstrict}} A
\;\;\;\;\;\text{and}\;\;\;\;\;
\lambda x. \lambda y. y \in (x:^{\qlazy} A) \xrightarrow{}
(y:^{\qstrict} A) \xrightarrow{\coeffcolor{x:\qlazy}} A.$$
The outer function type uses $x$ lazily in both cases
because the inner function suspends all usages
in its body; if, for example, we were to partially apply the first function,
this would not result in a strict usage of $x$ because the result
of that partial application would be a function that uses $x$ in its body.
But the types disagree
on the strictness of $y$
and on how the inner function uses $x$ in its body.

The difference makes these types
unsuitable for a boolean encoding.
With subsumption, however,
we can encode the Church booleans as $\lambda x. \lambda y. \textbf{sub } x$ and $\lambda x. \lambda y. \textbf{sub } y$
and assign both the less exact, but more useful, type%
\footnote{An effect-polymorphic extension of \cbnlangname could instead type both terms polymorphically.}
$$(x:^{\qlazy} A) \xrightarrow{}
(y:^{\qunknown} A) \xrightarrow{\coeffcolor{x:\qunknown}} A.$$

The other rules \ifextended \else (available in the extended version of this paper \cite{extended-version}) \fi are straightforward, keeping in mind that in \cbnlangname
all constructors are lazy.
Hence, all the introduction rules have an $\lazyvec$ effect
and store any latent effects on the type of the value they produce.

\section{Strictness in Call-By-Push-Value} \label{sec:details}

We would like to prove that the type system of \cbnlangname is sound
and yields the desired guarantees about usage.
However, carrying out these proofs directly in \cbnlangname is  awkward:
many language features in call-by-name languages can suspend computation,
so proofs about \cbnlangname will involve significant duplicated work.
In particular, the logical relations used in Sections \ref{subsec:semtyping} and \ref{sec:bot}
must enforce certain invariants relating types and values whenever a computation is thunked.
Instead of repeating these checks for every type, we would rather
carry out our metatheoretic
analysis in a language that isolates the thunking operation in
one syntactic construct with a dedicated type to represent it.
We can then translate \cbnlangname to this language in a type-preserving way.

Levy's call-by-push-value (CBPV) \cite{levy-call-by-push-value-1999,levy-call-by-push-value-2001,levy-cbpv-2003,Levy-cbpv-2022,levy-call-by-push-value-2006}
has precisely the features we need.
CBPV makes an explicit distinction between values and computations,
separating the two into different syntactic classes according to the slogan
``a value is, a computation does'' \cite{levy-cbpv-2003}.
In particular, it features a thunk value former,
written $\{M\}$, that suspends a computation $M$ as a value,
and a forcing operation, written $V!$, that forces a thunk $V$ by
running its suspended computation.
$\utype{} B$ is the type of thunks that produce computations of type $B$ when forced.
The thunking construct allows CBPV, despite being strictly evaluated,
to model lazily evaluated languages like \cbnlangname.

This section describes \cbpvlangname,
an extension of CBPV with strictness tracking,
and establishes the metatheoretic properties that make it a useful
model of intensional strictness.

\subsection{Syntax of \cbpvlangname} \label{subsec:syntax}

\begin{figure}

  \[
    \begin{array}{llcl}
      \coeffcolor{\textit{strictness attributes }} & \coeffsymbol & \defas & \qstrict \alt \qlazy \alt \qunknown \\
      \coeffcolor{\textit{attribute vectors }} & \coeffcolor{\gamma} & \defas & \coeffcolor{\cdot} \alt \coeffcolor{\gamma, x : \coeffsymbol} \\
      \textit{value types } & A &\defas & \texttt{unit} \alt \utype{\gamma} B \alt A_1 \times A_2 \alt A_1 + A_2 \\
      \textit{computation types } & B &\defas & A^{\coeffsymbol} \rightarrow B \alt \textbf{F} A \\
      \textit{contexts } & \Gamma &\defas& \cdot \alt \Gamma, x : A \\
      & & & \\
      \textit{values } & V &\defas &() \alt x \alt \{M\} \alt \textbf{inl } V \alt \textbf{inr } V \alt (V_1, V_2) \\
      \textit{computations } &M &\defas &\lambda x.\; M \alt M\;V \alt V! \alt  x \leftarrow M_1 \textbf{ in } M_2 \alt (x_1, x_2) \leftarrow V \textbf{ in } M \alt \textbf{sub } M \\
      & & &\alt\textbf{ret } V \alt V; M \alt \textbf{case } V \textbf{ of inl } x_1 \rightarrow M_1, \textbf{ inr } x_2 \rightarrow M_2
    \end{array}
  \]
  \vspace*{-0.5\baselineskip}
  \caption{Syntax of \cbpvlangname}
  \ifextended \vspace*{-1\baselineskip} \else \vspace*{-0.5\baselineskip} \fi
  \label{fig:CBPV-syntax}
\end{figure}

The primary difference between \cbpvlangname,
given in Figure \ref{fig:CBPV-syntax}, and standard presentations of CBPV
is the presence of an attribute vector $\coeffcolor{\gamma}$ on the
$\utype{}$ type former,
tracking how variables are used when a thunk is forced.
A variable annotated with an $\qlazy$ on a $\utype{}$ type will be
used lazily if a thunk inhabiting that type is forced, while
a variable with \qstrictarticle $\qstrict$ will be
used strictly. Function types also carry an attribute $\coeffsymbol$
describing how arguments are used when functions are applied.

Beyond the $\utype{}$ type, \cbpvlangname features two classes of types to
go with its two classes of terms.
{\em Positive types} (ranged over by $A$) describe values, while
{\em negative types} (written $B$) describe computations.
Value types include $\utype{}$ types along with the unit type, sums, and products.
Computation types include function types and $\ftype$ types, which are dual to $\utype{}$ types.
The \textit{returner type} $\ftype\;A$ describes computations returning values of type $A$.
As before, for technical expedience we extend the usual presentation of CBPV with a
\textbf{sub} term handling attribute subsumption.
\begin{comment}
The syntax of \cbpvlangname includes introduction and elimination forms for each of its types.
Most of these are fairly straightforward, but it is worth noting which positions
accept values and which accept computations.
Of particular note is the ``let'' construct (written $x \leftarrow M_1 \textbf{ in } M_2$),
which requires a computation (not a value) in the position being bound.
The type system requires this computation to have an $\ftype$ type---that is,
the computation $M_1$ must be one that returns a value.
In this way ``let'' behaves like a bind and serves as the elimination form for the
$\ftype$ type.
\end{comment}

\subsection{Typing Rules for \cbpvlangname} \label{subsec:typing}

\cbpvlangname has two typing judgments:
$\coeffcolor{\gamma} \cdot \Gamma \vdash V : A$ for values and
$\coeffcolor{\gamma} \cdot \Gamma \vdash M : B$ for computations.
The shapes of these judgments are somewhat different from the \cbnlangname typing judgment;
instead of appearing on the~ $:$ like an effect, $\coeffcolor{\gamma}$
is zipped with the context $\Gamma$ on the left of the $\vdash$.\footnote{
  Some readers may notice a similarity to coeffect systems in the judgment's shape;
we discuss this similarity in Section \ref{subsec:naming}. }
%, which is more typical for a graded coeffect \cite{abel-coeffects,choudhury-graded-2021}.
We can zip and unzip $\coeffcolor{\gamma} \cdot \Gamma$ freely, i.e., we write
$\coeffcolor{\gamma} \cdot \Gamma, x :^{\coeffsymbol} A$ interchangeably with $\coeffcolor{(\gamma, x:\coeffsymbol)} \cdot (\Gamma, x: A)$.
This is different from how \cbnlangname contexts work,
as those associate variables with both types and $\coeffcolor{\gamma}$s.
This difference reflects the fact that \cbpvlangname is evaluated strictly,
rather than lazily like \cbnlangname, so
the use of a variable in \cbpvlangname does not itself trigger any evaluation.
Instead, this triggering occurs in the thunking and forcing rules,
which suspend and resume computation.
Accordingly, whereas in \cbnlangname all type constructors must be graded with attributes
because all type constructors suspend computation,
in CBPV only the $\utype{}$ type is graded because only the $\utype{}$ type suspends computation.
We can think of \cbnlangname types as implicitly containing  $\utype{}$ types
(as is made explicit by translation in Figure \ref{fig:cbn-trans}),
and therefore all \cbnlangname types need to keep track of strictness attributes.

Note also that the definition of contexts requires that variables be bound to value types;
binding a computation to a variable requires it to be explicitly thunked.

As we consider the rules, recall the second example from Section \ref{subsec:examples},
where we described how \cbnlangname breaks the usual $\eta$-equivalence
between $x$ and $( \textbf{fst } x, \textbf{snd } x)$.
The pair constructor in \cbnlangname implicitly
thunks its two arguments;
this is the reason why $x$ is lazy in the latter term but not the former.
To capture this behavior in \cbpvlangname, the thunk value
must change the attributes associated with the computation it suspends,
breaking the $\eta$-equivalence between $V$ and $\{V!\}$.
In particular, we want \cbpvlangname to describe the term $x$ as using
the variable $x$ strictly
but $\{x!\}$ as using it lazily.

\begin{figure}
  \begin{center}
    \begin{tabularx}{1\textwidth}{
        >{\raggedright\arraybackslash}X
        >{\raggedleft\arraybackslash}X
      }
      \fbox{$\coeffcolor{\gamma}\cdot\Gamma \vdash V : A$} & \textit{(Value typing)}
    \end{tabularx}
  \end{center}
  $$
  \infer[\rlabel*{T-Var}]{\lazyvec \cdot \Gamma_1, x :^{\qstrict} A, \lazyvec \cdot \Gamma_2 \vdash x : A}{}
  \qquad
  \infer[\rlabel*{T-Thunk}]{\lazyvec \cdot \Gamma \vdash \{ M \} : \utype{\gamma} B}
  {\coeffcolor{\gamma} \cdot \Gamma \vdash M : B}
  $$

  \vspace{3mm}

  \begin{center}
    \begin{tabularx}{1\textwidth}{
        >{\raggedright\arraybackslash}X
        >{\raggedleft\arraybackslash}X
      }
      \fbox{$\coeffcolor{\gamma}\cdot\Gamma \vdash M : B$} & \textit{(Computation typing)}
    \end{tabularx}
  \end{center}

  $$
  \infer[\rlabel*{T-Return}]
  {\coeffcolor{\gamma} \cdot \Gamma \vdash \textbf{ret } V : \ftype A}
  {\coeffcolor{\gamma} \cdot \Gamma \vdash V : A} \qquad
  \infer[\rlabel*{T-Sub}]
  {\coeffcolor{\gamma} \cdot \Gamma \vdash \textbf{sub } M : B}
  {\coeffcolor{\gamma'} \cdot \Gamma \vdash M : B & \coeffcolor{\gamma \leq \gamma'}}
  \vspace{2.75ex}
  $$
  $$
  \infer[\rlabel*{T-Force}]{\coeffcolor{\gamma_1 + \gamma_2} \cdot \Gamma \vdash V! :  B}
  {\coeffcolor{\gamma_1} \cdot \Gamma \vdash V : \utype{\gamma_2} B} \qquad
  \infer[\rlabel*{T-Let}]
  {\coeffcolor{(\gamma_1 + \gamma_2)} \cdot \Gamma \vdash
  x \leftarrow M_1 \textbf{ in } M_2 : \coeffcolor{\downshift{x}} B}
  {\coeffcolor{\gamma_1} \cdot \Gamma \vdash M_1 : \ftype A &
    \coeffcolor{\gamma_2} \cdot \Gamma, x :^{\coeffsymbol} A \vdash M_2 : B &
  }
  \vspace{2.75ex}
  $$
  $$
  \infer[\rlabel*{T-Seq}]
  {\coeffcolor{(\gamma_1 + \gamma_2)} \cdot \Gamma \vdash V; M : B}
  {\coeffcolor{\gamma_1} \cdot \Gamma \vdash V : \texttt{unit} &
    \coeffcolor{\gamma_2} \cdot \Gamma \vdash M : B &
  }
  \vspace{2.75ex}
  $$
  $$
  \infer[\rlabel*{T-Abs}]
  {\coeffcolor{\gamma} \cdot \Gamma \vdash \lambda x. M : A^{\coeffsymbol} \rightarrow \coeffcolor{\downshift{x}} B}
  {\coeffcolor{\gamma} \cdot \Gamma, x :^{\coeffsymbol} A \vdash M : B} \qquad
  \infer[\rlabel*{T-App}]
  {\coeffcolor{\gamma_1 + \gamma_2} \cdot \Gamma \vdash M\;V :  B}
  {\coeffcolor{\gamma_1} \cdot \Gamma \vdash M :A^{\coeffsymbol} \rightarrow B&
    \coeffcolor{\gamma_2} \cdot \Gamma \vdash V : A
  }
  $$
  \caption{Main typing rules for \cbpvlangname}
  \vspace*{-0.75\baselineskip}
  \label{fig:cbpv-typing-important}
\end{figure}

To type $x$ properly, we must ensure that
variable lookup produces \qstrictarticle $\qstrict$ attribute.
This leads directly to the \rref{T-Var} rule in Figure \ref{fig:cbpv-typing-important}.
The rule requires that only the variable being used has a strict attribute,
isolating any imprecision in \cbpvlangname to the \rref{T-Sub} rule.

To type $\{x!\}$,
we need the thunking operation to make the use of $x$ inside the thunk lazy.
Accordingly, the \rref{T-Thunk} rule suspends all
the uses of the variables in $M$ represented by $\coeffcolor{\gamma}$ and packages them into the $\utype{}$ type,
to be released later if the thunk is forced.
The vector of attributes on the thunk itself is then set to the lazy
$\lazyvec$ vector, since
a thunk value uses nothing strictly.
So, for example, the value $\{\textbf{ret } x\}$
has the type $\utype{x:\qstrict} \textbf{ } \ftype \textbf{ } A$ when $x :^{\qlazy} A \in \coeffcolor{\gamma} \cdot\Gamma$.

The \rref{T-Force} rule is more complex.
Intuitively, the vector of attributes $\coeffcolor{\gamma_1}$ is
used to produce $V$, which has the type $\utype{\gamma_2} B$.
This $\coeffcolor{\gamma_2}$ represents the usages that will occur when $V$ is forced,
so those attributes are ``replayed'' here into the context by adding them
to $\coeffcolor{\gamma_1}$.
So, were we to add an $!$ to our example from before to produce the computation $\{\textbf{ret } x\}!$,
this new program would have the type $\ftype \; A$ when $x :^{\qstrict} A \in \coeffcolor{\gamma} \cdot\Gamma$.
The attributes from the thunk's type have been added
to those used to produce it (in this case $\lazyvec$).

With these three rules,
we can see how \cbpvlangname differentiates $x$ from $\{x!\}$.
Assuming $x$ is bound to a thunk value of type
$\utype{} B$, \cbpvlangname types the two programs with the derivations
$$
x:^{\qstrict} \utype{} B \vdash x : \utype{\coeffcolor{x:\qlazy}} B
\;\;\;\;\;\;\;\;\;\text{and}\;\;\;\;\;\;\;\;\;
x:^{\qlazy} \utype{} B \vdash \{x!\} : \utype{\coeffcolor{x:\qstrict}} B
$$
respectively.
The strict use of $x$ has been moved from the context in the former program onto the
$\utype{}$ type in the latter, describing how the use of $x$
has been suspended by the thunk.

The \rref{T-Abs} and \rref{T-App} rules handle function type introduction and elimination.
Lambda abstraction does not suspend the attributes of the body,
so the $\coeffcolor{\gamma}$ necessary to check the $M$ in $\lambda x. M$ ``passes through''
the function (e.g., $\lambda x. \textbf{ ret } y$ is considered to use $y$ strictly).
Instead, since functions need to
be thunked to be bound to variables,
their attributes will be suspended by the thunk rule.
This treatment of attributes in functions may appear surprising,
but it is common practice in the literature on effects and coeffects in CPBV
\cite{torczon-effects-2024, kammar-2013,kammar-2012}
and has the benefit of
isolating all reasoning about the suspension and resumption of attributes to the
rules for the $\utype{}$ type.
This behavior also means that, unlike \cbnlangname,
\cbpvlangname's arrows do not include their argument in the scope
of their return type.
Consider, for example, the function $\lambda x.\; \lambda y.\; \textbf{ret }x$.
The inner function uses $x$ strictly, and, because this strict usage
passes through the inner function to the outer, we can assign the latter the type
$A_1^{\qstrict} \rightarrow A_2^{\qlazy} \rightarrow \ftype\;A_1$.
Accordingly, when producing the result type in the
\rref{T-Abs} rule, we remove $x$ from the arrow's return type $B$.

For a more interesting example, a function $\lambda x. (x; \textbf{ret } y)$ can be typed as
$\texttt{unit}^{\coeffcolor{\qstrict}} \rightarrow \ftype \text{ } A$ in a context where $y$ has the type $A$ and attribute $\qstrict$,
since it uses both its argument and $y$ strictly.
Conversely, a function $\lambda x.\; \textbf{ret } \{ x; \textbf{ret } y \}$ can be assigned type
$\texttt{unit}^{\coeffcolor{\qlazy}} \rightarrow  \ftype \text{ } \utype{y:\qstrict} \;\ftype A$,
since it uses its argument and $y$ lazily and produces a thunk that
will use $y$ strictly if forced.
As discussed above, the usage of $x$ is not mentioned in the return type of the function;
an external caller only cares about how the result of the function uses
the variables that exist in the calling scope.

The \rref{T-Let} rule has a different structure from its \cbnlangname counterpart,
since \cbpvlangname evaluates strictly while \cbnlangname does not.
Unlike \rref{T-CBN-Let}, the vector $\coeffcolor{\gamma_1}$ used in the evaluation of $M_1$ is not implicitly
present in the context for the evaluation of $M_2$,
so the \rref{T-Let} rule must explicitly add it to the vector $\coeffcolor{\gamma_2}$
used by $M_2$ to yield the resulting attributes for the entire term.
The rule otherwise behaves like the \rref{T-CBN-Let} rule,
including in its handling of scoping.

To better understand this scoping,
consider the program
$
z \leftarrow (x \leftarrow \textbf{ret } () \textbf{ in } \textbf{ret } \{ x; \textbf{ret } y \}) \textbf{ in } z!
$,
assuming $y$ has type $A$.
We would like this program to have the type $\ftype \text{ } A$ and assign an $\qstrict$ attribute to $y$.
At the beginning of the program, the variable $x$ is assigned a $()$ value, and hence will type at $\texttt{unit}$.
Thus, the thunk value $\{ x; \textbf{ret } y \}$ has the type
$\utype{x:\qstrict, y:\qstrict} \textbf{ } \ftype \textbf{ } A$ when $x$ and $y$ have
an $\qlazy$ attribute.
However, by the time this thunk is forced, $x$ has gone out of scope.
Bearing in mind our goal of assigning $\qstrict$ to $y$, however, we can see that $x$'s usage inside the
thunk is irrelevant to $y$;
we can forget $x$ from the type
$\utype{x:\qstrict, y:\qstrict} \textbf{ } \ftype \textbf{ } A$ to yield
$\utype{y:\qstrict} \textbf{ } \ftype \textbf{ } A$,
which has exactly the behavior we want.

\newcommand{\figcbpvtypingextra}{

  \begin{center}
    \begin{tabularx}{1\textwidth}{
        >{\raggedright\arraybackslash}X
        >{\raggedleft\arraybackslash}X
      }
      \fbox{$\coeffcolor{\gamma}\cdot\Gamma \vdash V : A$} & \textit{(Value typing)}
    \end{tabularx}
  \end{center}

  $$
  \infer[\rlabel*{T-Unit}]
  {\lazyvec \cdot \Gamma \vdash () : \texttt{unit}}
  {} \qquad
  \infer[\rlabel*{T-Pair}]
  {\coeffcolor{(\gamma_1 + \gamma_2)} \cdot \Gamma \vdash (V_1, V_2) : A_1 \times A_2}
  {
    \coeffcolor{\gamma_1} \cdot \Gamma \vdash V_1 : A_1 &
    \coeffcolor{\gamma_2} \cdot \Gamma \vdash V_2 : A_2 &
  }
  \vspace{2.75ex}
  $$
  $$
  \infer[\rlabel*{T-Inl}]
  {\coeffcolor{\gamma} \cdot \Gamma \vdash \textbf{inl } V : A_1 + A_2}
  {\coeffcolor{\gamma} \cdot \Gamma \vdash V : A_1} \qquad
  \infer[\rlabel*{T-Inr}]
  {\coeffcolor{\gamma} \cdot \Gamma \vdash \textbf{inr } V : A_1 + A_2}
  {\coeffcolor{\gamma} \cdot \Gamma \vdash V : A_2}
  $$

  \vspace{5mm}

  \begin{center}
    \begin{tabularx}{1\textwidth}{
        >{\raggedright\arraybackslash}X
        >{\raggedleft\arraybackslash}X
      }
      \fbox{$\coeffcolor{\gamma}\cdot\Gamma \vdash M : B$} & \textit{(Computation typing)}
    \end{tabularx}
  \end{center}

  $$
  \infer[\rlabel*{T-Split}]
  {\coeffcolor{(\gamma_1 + \gamma_2)} \cdot \Gamma \vdash
  (x_1, x_2) \leftarrow V \textbf{ in } M : \coeffcolor{\downshift{x_1}\downshift{x_2}} B}
  {\coeffcolor{\gamma_1} \cdot \Gamma \vdash V : A_1 \times A_2 &
    \coeffcolor{\gamma_2} \cdot \Gamma, x_1:^{\coeffcolor{\coeffsymbol_1}} A_1, x_2:^{\coeffcolor{\coeffsymbol_2}} A_2 \vdash M : B &
  }
  \vspace{2.75ex}
  $$
  $$
  \infer[\rlabel*{T-Case}]
  {\coeffcolor{(\gamma_1 + \gamma_2)} \cdot \Gamma \vdash
  \textbf{case } V \textbf{ of inl } x_1 \rightarrow M_1, \textbf{ inr } x_2 \rightarrow M_2: B}
  {\deduce
    {\deduce
      {\deduce
        {}
        {\coeffcolor{\gamma_2} \cdot \Gamma, x_1:^{\coeffcolor{\coeffsymbol_1}} A_1 \vdash M_1 : B_1 \;\;&\;\; \coeffcolor{\gamma_2} \cdot \Gamma, x_2:^{\coeffcolor{\coeffsymbol_2}} A_2 \vdash M_2 : B_2}
      }
      {}
    }
    {\coeffcolor{\gamma_1} \cdot \Gamma \vdash V : A_1 + A_2 \;\;&\;\; \coeffcolor{B = \downshift{x_1} B_1 = \downshift{x_2} B_2}}
  }
  $$
}

\ifextended
\begin{figure}
  \figcbpvtypingextra
  \caption{Additional typing rules for \cbpvlangname}
  \label{fig:cbpv-misc-typing}
\end{figure}
\fi
The other rules are straightforward;
they can be found in \ifextended Figure \ref{fig:cbpv-misc-typing}. \else the extended version of this paper  \cite{extended-version}. \fi

\begin{figure}
  \[
    \begin{array}{llcl}
      \textit{closed terminal values } &W &\defas&() \alt (W, W) \alt \textbf{inl } W \alt \textbf{inr } W \alt \{\coeffcolor{\gamma}, \rho, M\}\\
      \textit{closed terminal computations } &T &\defas &\textbf{ret } W \alt \llangle \coeffcolor{\gamma}, \rho, \lambda x. M \rrangle \\
      \textit{environments } &\rho&\defas &\cdot \alt \rho, x \mapsto W
    \end{array}
  \]
  \ifextended \vspace*{-.5\baselineskip} \else \vspace*{-0.5\baselineskip} \fi
  \caption{Closed terminal values, computations and environments in \cbpvlangname}
  \label{fig:cbpv-values}
  \ifextended \vspace*{-1\baselineskip}  \else \vspace*{-0.75\baselineskip} \fi
\end{figure}

\subsection{Big-Step Semantics of \cbpvlangname} \label{subsec:semantics}

To prove that evaluation of \cbpvlangname reflects the strictness attributes computed by the type system,
we enrich the standard semantics of \cbpvlangname with strictness-attribute tracking.
We present the semantics in a big-step style,
defining terminal values $W$ and computations $T$ in Figure \ref{fig:cbpv-values}, along with environments $\rho$.
Like $\Gamma$s, $\rho$s can freely be zipped and unzipped with $\coeffcolor{\gamma}$s.

\cbpvlangname's two evaluation judgments have the form $\coeffcolor{\gamma} \cdot \rho \vdash V \Downarrow W$ for values
and $\coeffcolor{\gamma} \cdot \rho \vdash M \Downarrow T$ for computations. The rules
of particular relevance to strictness tracking are given in Figure \ref{fig:cbpv-semantics-important}.

\begin{figure}

  \begin{center}
    \begin{tabularx}{1\textwidth}{
        >{\raggedright\arraybackslash}X
        >{\raggedleft\arraybackslash}X
      }
      \fbox{$\coeffcolor{\gamma}\cdot\rho \vdash V \Downarrow W$} & \textit{(Value semantics)}
    \end{tabularx}
  \end{center}

  $$
  \infer[\rlabel*{E-Var}]
  {\lazyvec \cdot \rho_1, x \mapsto^{\qstrict} W, \lazyvec \cdot \rho_2 \vdash x \Downarrow W}
  {} \qquad
  \infer[\rlabel*{E-Thunk}]
  {\lazyvec \cdot \rho \vdash \{ M \} \Downarrow \{ \coeffcolor{\gamma}, \rho , M \}}
  {}
  $$

  \vspace{3mm}

  \begin{center}
    \begin{tabularx}{1\textwidth}{
        >{\raggedright\arraybackslash}X
        >{\raggedleft\arraybackslash}X
      }
      \fbox{$\coeffcolor{\gamma}\cdot\rho \vdash M \Downarrow T$} & \textit{(Computation semantics)}
    \end{tabularx}
  \end{center}

  $$
  \infer[\rlabel*{E-Sub}]
  {\coeffcolor{\gamma} \cdot \rho \vdash \textbf{sub } M \Downarrow T}
  {\coeffcolor{\gamma'} \cdot \rho \vdash M \Downarrow T & \coeffcolor{\gamma \leq \gamma'}} \qquad
  \infer[\rlabel*{E-Return}]
  {\coeffcolor{\gamma} \cdot \rho \vdash \textbf{ret } V \Downarrow \textbf{ret } W}
  {\coeffcolor{\gamma} \cdot \rho \vdash V \Downarrow W}
  \vspace{2.75ex}
  $$
  $$
  \infer[\rlabel*{E-Let}]
  {\coeffcolor{(\gamma_1 + \gamma_2)} \cdot \rho \vdash
  x \leftarrow M_1 \textbf{ in } M_2 \Downarrow T}
  {\deduce{\coeffcolor{\gamma_2} \cdot \rho, x \mapsto^{\coeffsymbol} W \vdash M_2 \Downarrow T}
    {\coeffcolor{\gamma_1} \cdot \rho \vdash M_1 \Downarrow \textbf{ret } W}
  } \qquad
  \infer[\rlabel*{E-Force}]
  {\coeffcolor{\gamma_1 + (\gamma_2 \mid_{\dom{\gamma_1}})} \cdot \rho \vdash V! \Downarrow  T}
  {\deduce {\coeffcolor{\gamma_2} \cdot \rho' \vdash M \Downarrow T}
    {\coeffcolor{\gamma_1} \cdot \rho \vdash V \Downarrow \{\coeffcolor{\gamma_2}, \rho', M \}}
  }
  \vspace{2.75ex}
  $$
  $$
  \infer[\rlabel*{E-Seq}]
  {\coeffcolor{(\gamma_1 + \gamma_2)} \cdot \rho \vdash V; M \Downarrow T}
  {\coeffcolor{\gamma_1} \cdot \rho \vdash V \Downarrow () &
    \coeffcolor{\gamma_2} \cdot \rho \vdash M \Downarrow T &
  } \qquad
  \infer[\rlabel*{E-Abs}]
  {\coeffcolor{\gamma} \cdot \rho \vdash \lambda x. M \Downarrow
  \llangle \coeffcolor{\gamma}, \rho, \lambda x. M \rrangle}
  {}
  \vspace{2.75ex}
  $$
  $$
  \infer[\rlabel*{E-App}]
  {\coeffcolor{\gamma_1 + \gamma_2} \cdot \rho \vdash M\;V \Downarrow T}
  { \coeffcolor{\gamma_1} \cdot \rho \vdash M \Downarrow \llangle \coeffcolor{\gamma_3}, \rho', \lambda x. M' \rrangle &
    \coeffcolor{\gamma_2} \cdot \rho \vdash V \Downarrow W  &
    \coeffcolor{\gamma_3} \cdot \rho', x \mapsto^{\coeffsymbol} W \vdash M' \Downarrow T
  }
  $$
  \caption{Main semantic rules for \cbpvlangname}
  \ifextended \else \vspace*{-0.5\baselineskip} \fi
  \label{fig:cbpv-semantics-important}
  \ifextended \vspace*{-\baselineskip} \else \fi
\end{figure}

\newcommand{\figcbpvmiscsemantics}{

  \begin{center}
    \begin{tabularx}{1\textwidth}{
        >{\raggedright\arraybackslash}X
        >{\raggedleft\arraybackslash}X
      }
      \fbox{$\coeffcolor{\gamma}\cdot\rho \vdash V \Downarrow W$} & \textit{(Value semantics)}
    \end{tabularx}
  \end{center}

  $$
  \infer[\rlabel*{E-Unit}]
  {\lazyvec \cdot \rho \vdash () \Downarrow ()}
  {} \qquad
  \infer[\rlabel*{E-Pair}]
  {\coeffcolor{(\gamma_1 + \gamma_2)} \cdot \rho \vdash (V_1, V_2) \Downarrow (W_1, W_2)}
  {
    \coeffcolor{\gamma_1} \cdot \rho \vdash V_1 \Downarrow W_1 &
    \coeffcolor{\gamma_2} \cdot \rho \vdash V_2 \Downarrow W_2
  }
  \vspace{2.75ex}
  $$
  $$
  \infer[\rlabel*{E-Inl}]
  {\coeffcolor{\gamma} \cdot \rho \vdash \textbf{inl } V \Downarrow \textbf{inl } W }
  {\coeffcolor{\gamma} \cdot \rho \vdash V \Downarrow W} \qquad
  \infer[\rlabel*{E-Inr}]
  {\coeffcolor{\gamma} \cdot \rho \vdash \textbf{inr } V \Downarrow \textbf{inr } W }
  {\coeffcolor{\gamma} \cdot \rho \vdash V \Downarrow W}
  $$

  \vspace{5mm}

  \begin{center}
    \begin{tabularx}{1\textwidth}{
        >{\raggedright\arraybackslash}X
        >{\raggedleft\arraybackslash}X
      }
      \fbox{$\coeffcolor{\gamma}\cdot\rho \vdash M \Downarrow T$} & \textit{(Computation semantics)}
    \end{tabularx}
  \end{center}

  $$
  \infer[\rlabel*{E-Split}]
  {\coeffcolor{(\gamma_1 + \gamma_2)} \cdot \rho \vdash
  (x_1, x_2) \leftarrow V \textbf{ in } M \Downarrow T}
  {\coeffcolor{\gamma_1} \cdot \rho \vdash V \Downarrow (W_1, W_2) &
    \coeffcolor{\gamma_2} \cdot \rho, x_1\mapsto^{\coeffcolor{\coeffsymbol_1}} W_1, x_2\mapsto^{\coeffcolor{\coeffsymbol_2}} W_2 \vdash M \Downarrow T &
  }
  \vspace{2.75ex}
  $$
  $$
  \infer[\rlabel*{E-Case-Inl}]
  {\coeffcolor{\gamma_1 + \gamma_2} \cdot \rho \vdash
  \textbf{case } V \textbf{ of inl } x_1 \rightarrow M_1, \textbf{ inr } x_2 \rightarrow M_2 \Downarrow T}
  {\coeffcolor{\gamma_1} \cdot \rho \vdash V \Downarrow \textbf{inl } W&
  \coeffcolor{\gamma_2} \cdot \rho, x_1 \mapsto^{\coeffsymbol} W \vdash M_1 \Downarrow T}
  \vspace{2.75ex}
  $$
  $$
  \infer[\rlabel*{E-Case-Inr}]
  {\coeffcolor{\gamma_1 + \gamma_2} \cdot \rho \vdash
  \textbf{case } V \textbf{ of inl } x_1 \rightarrow M_1, \textbf{ inr } x_2 \rightarrow M_2 \Downarrow T}
  {\coeffcolor{\gamma_1} \cdot \rho \vdash V \Downarrow \textbf{inr } W &
  \coeffcolor{\gamma_2} \cdot \rho, x_2 \mapsto^{\coeffsymbol} W \vdash M_2 \Downarrow T}
  $$
}

\ifextended
\begin{figure}
  \figcbpvmiscsemantics
  \caption{Additional semantic rules for \cbpvlangname}
  \vspace{-1\baselineskip}
  \label{fig:cbpv-misc-semantic}
\end{figure}
\fi

Most of the complexity in the semantics of \cbpvlangname arises from how thunk
values handle the scopes of their captured attribute vectors.
In the $\rref{E-Thunk}$ rule, where thunk values are created,
the existing environment is captured and a $\coeffcolor{\gamma}$ is chosen for the result.
This rule is unusual in that it allows the choice to be any
$\coeffcolor{\gamma}$---to see why, remember that we are adding
$\coeffcolor{\gamma}$ to the semantics
only to enable metatheoretic reasoning;
the semantics here is free to select whichever $\coeffcolor{\gamma}$ is necessary to
evaluate this thunk, should it later be forced.
However, by the time a thunk \textit{is} forced,
the scope of the vector it captured may no longer be the same as the
scope in which it is being forced.
% \bcp{Not sure I got that. But maybe the next bit helps...}

Concretely, in the $\rref{E-Force}$ rule, there are two different scopes in play.
The outer scope, where $\coeffcolor{\gamma_1}$ and $\rho$ live,
is the same scope in which the typing judgment will view a force computation.
The inner scope, where $\coeffcolor{\gamma_2}$ and $\rho'$ live, is internal to the thunk value
and may be arbitrarily different from the outer scope.
Accordingly, the scope of $\coeffcolor{\gamma_2}$ must be adjusted when adding it to $\coeffcolor{\gamma_1}$ in the rule's conclusion;
this adjustment is exactly the restriction of $\coeffcolor{\gamma_2}$ to the domain of $\coeffcolor{\gamma_1}$
(written $\coeffcolor{\gamma_2 \mid_{\dom{\gamma_1}}}$).

As an example, recall the program from earlier:
$
z \leftarrow (x \leftarrow \textbf{ret } () \textbf{ in } \textbf{ret } \{ x; \textbf{ret } y \}) \textbf{ in } z!
$.
Assuming $\rho$ maps $y$ to some value $W$,
the bound term $\{ x; \textbf{ret } y \}$ reduces to a thunk value capturing its environment:
$\{\coeffcolor{(y : \qstrict, x: \qstrict)}, (x \mapsto (), y \mapsto W), x; \textbf{ret } y \}$.
When that value is forced, however, only $y$ is in scope,
so we take the restriction of $\coeffcolor{y : \qstrict, x: \qstrict}$ (i.e., $\coeffcolor{y : \qstrict}$)
and add it to the attributes used to produce the thunk ($\lazyvec$)
to get the result $\coeffcolor{y : \qstrict}$, which agrees with the typing rules, as desired.
Intuitively, $x$ is local to the thunk here, so
we do not care about its usage when forcing the thunk.

It is also worth pointing out the $\rref{E-App}$ rule,
which mirrors prior work in CBPV but is potentially confusing nonetheless.
As mentioned previously, \cbpvlangname functions do not suspend their attributes
(also seen here in the $\rref{E-Abs}$ rule).
Therefore, the application rule does not need to include
the attributes $\coeffcolor{\gamma_3}$ from the closure in the result of the application,
as they will already be included in the attributes $\coeffcolor{\gamma_1}$
needed to produce the closure in the first place.

The other rules are straightforward;
they can be found in
\ifextended Figure \ref{fig:cbpv-misc-semantic}.
\else the extended version of this paper  \cite{extended-version}. \fi

\subsection{Soundness of \cbpvlangname} \label{subsec:semtyping}

We can now prove that the attributes carried through the semantics
reflect the attributes computed by the typing judgment.
The proofs in this section, along with all those that follow,
have been mechanized in Rocq \cite{artifact}.
These proofs require a surprising amount of scope bookkeeping
and are most tractable with a logical relation.
This relation, presented in Figure \ref{fig:lr-semtyping},
resembles the one used by \citet{torczon-effects-2024} to prove coeffect soundness,
differing only in the thunk and function cases.

\newcommand{\lrv}[1]{\mathcal{W}\llbracket #1 \rrbracket}
\newcommand{\lrc}[1]{\mathcal{T}\llbracket #1 \rrbracket}
\newcommand{\lrm}[1]{\mathcal{M}\llbracket #1 \rrbracket}
\newcommand{\lrf}[1]{\mathcal{F}\llbracket #1 \rrbracket}
\newcommand{\set}[1]{\{#1\}}
\newcommand{\andtext}{\textit{\;and\;}}
\newcommand{\foralltext}{\textit{\;for\;all\;}}
\newcommand{\impliestext}{\;\;\implies\;\;}
\newcommand{\existstext}{\textit{\;there\;exists\;}}
\newcommand{\ortext}{\textit{\;or\;}}

In the thunk case,
we need to handle the fact that the scopes in the $\utype{}$ type and
in the thunk value may have become ``misaligned'';
that is, variables may have been introduced to (or removed from) the
type that are not present (or are still present) in the thunk value.
However, recalling the earlier examples we used when defining the evaluation rule \rref{E-Force},
we see that strictness tracking for values is only concerned with the attributes that
exist in the scope using a value, not the scope that defined it.
Thus, for the purposes of semantic typing,
we care only to check that the attributes in the value agree with those in the type,
when restricted to the domain that they share.\footnote{
  The mechanization of this proof uses well-scoped de Bruijn syntax, and as a result
  needs somewhat more elaborate bookkeeping in this case, since in a nameless setting it is
  not possible to determine the overlap between two arbitrary scopes.
  Instead, the scopes of values must be tracked explicitly and adjusted as they flow through
  the program.
More details about this can be found in the artifact associated with this paper \cite{artifact}.}

\begin{figure}
  \begin{align*}
    \lrv{\texttt{unit}} &\;\;=\;\; \set{()} \\
    \lrv{A_1 \times A_2} &\;\;=\;\; \set{(W_1, W_2) \alt W_1 \in \lrv{A_1} \andtext W_2 \in \lrv{A_2}} \\
    \lrv{A_1 + A_2} &\;\;=\;\; \set{\textbf{inl } W_1 \alt W_1 \in \lrv{A_1}} \cup \set{\textbf{inr } W_2 \alt W_2 \in \lrv{A_2}} \\
    \lrv{\utype{\gamma} B} &\;\;=\;\; \set{\{\coeffcolor{\gamma'}, \rho, M \} \alt
      \coeffcolor{\gamma \mid_{\dom{\gamma'}} \;=\; \gamma' \mid_{\dom{\gamma}}} \andtext
    \coeffcolor{\gamma'} \cdot \rho \vdash M \Downarrow T \andtext T \in \lrc{B}} \\
    \lrc{\ftype A } &\;\;=\;\; \set{\textbf{ret } W \alt W \in \lrv{A}} \\
    \lrc{A^{\coeffsymbol} \rightarrow B} &\;\;=\;\; \set{\llangle \coeffcolor{\gamma}, \rho, \lambda x. M \rrangle \alt
      \foralltext W \in \lrv{A},
    \coeffcolor{\gamma} \cdot \rho, x \mapsto^{\coeffsymbol} W \vdash M \Downarrow T \andtext T \in \lrc{B} }
  \end{align*}
  \begin{align*}
    \Gamma \vDash \rho &\;\;\;\triangleq\;\;\; x: A \in \Gamma \impliestext x \mapsto W \in \rho \andtext W \in \lrv{A} \\
    \coeffcolor{\gamma} \cdot \Gamma \vDash V : A &\;\;\;\triangleq\;\;\; \Gamma \vDash \rho \impliestext \coeffcolor{\gamma}  \cdot \rho \vdash V \Downarrow W \andtext W \in \lrv{A} \\
    \coeffcolor{\gamma}  \cdot \Gamma \vDash M : B &\;\;\;\triangleq\;\;\; \Gamma \vDash \rho \impliestext \coeffcolor{\gamma}  \cdot \rho \vdash M \Downarrow T \andtext T \in \lrc{B}
  \end{align*}
  \vspace*{-1.35\baselineskip}
  \caption{Logical Relation and Semantic Typing for Soundness}
  \label{fig:lr-semtyping}
  \vspace*{-0.65\baselineskip}
\end{figure}

With this logical relation in hand,
we can define the usual notion of semantic typing in Figure \ref{fig:lr-semtyping}.
We state and prove the fundamental lemma for this relation and derive
soundness as a corollary.

\begin{lemma}[Fundamental Lemma: Soundness]
  \label{lem:fundamental-lemma}
  For all $ \coeffcolor{\gamma} $ and $\Gamma$, if $ \coeffcolor{\gamma}  \cdot \Gamma \vdash V : A$, then $ \coeffcolor{\gamma}  \cdot \Gamma \vDash V : A$,
  and if $ \coeffcolor{\gamma}  \cdot \Gamma \vdash M : B$, then $ \coeffcolor{\gamma}  \cdot \Gamma \vDash M : B$.
\end{lemma}

\begin{proof}
  By mutual induction on the typing derivations for values and computations.
\end{proof}

\begin{theorem}[Soundness]
  \label{thm:soundness}
  Given $\coeffcolor{\gamma} $, $\Gamma$, and $\rho$ such that $\Gamma \vDash \rho$,
  \begin{enumerate}
    \item if $\coeffcolor{\gamma}  \cdot \Gamma \vdash V : A$, then there exists some $W$ such that $\coeffcolor{\gamma} \cdot \rho \vdash V \Downarrow W$;
    \item if $\coeffcolor{\gamma}  \cdot \Gamma \vdash M : B$, then there exists some $T$ such that $\coeffcolor{\gamma} \cdot \rho \vdash M \Downarrow T$.
  \end{enumerate}
\end{theorem}

\begin{proof}
  Follows directly from the fundamental lemma.
\end{proof}

We can also show that the
$\coeffsymbol$s we place on arrow types accurately describe the strictness of function arguments.
In other words, if a computation has a function type,
then it must evaluate to a terminal closure. This closure, furthermore, will successfully evaluate
in an extended environment where the $\coeffsymbol$ on the function type
marks the strictness of the argument.

\begin{theorem}[Function Type Soundness]
  \label{thm:fn-soundness}
  Given $\coeffcolor{\gamma} $, $\Gamma$, $A$, $B$, $\coeffsymbol$, and $\rho$ such that $\Gamma \vDash \rho$,
  if $\coeffcolor{\gamma}  \cdot \Gamma \vdash M : A^{\coeffsymbol} \rightarrow B$,
  then there exists some $\coeffcolor{\gamma'}$, $\rho'$ and $M'$ such that
  $\coeffcolor{\gamma} \cdot \rho \vdash M \Downarrow \llangle \coeffcolor{\gamma'}, \rho', \lambda x. M' \rrangle$.
  Additionally, for any $W$ such that $W \in \lrv{A}$,
  there exists some $T$ such that $\coeffcolor{\gamma'} \cdot \rho', x \mapsto^{\coeffsymbol} W \vdash M' \Downarrow T$.
\end{theorem}

\begin{proof}
  Follows directly from the fundamental lemma.
\end{proof}

\subsection{\cbnlangname Translation}\label{subsec:translation}

Having proved the soundness of \cbpvlangname,
we argue that it can be used as a model of strictness
in \cbnlangname, the call-by-name system from Section \ref{sec:cbn}.
We give semantics to
\cbnlangname via Levy's
standard translation \cite{Levy-cbpv-2022}
from the CBN lambda calculus to CBPV,
which we extend to \cbnlangname and \cbpvlangname.

\newcommand{\translate}[1]{\llbracket #1 \rrbracket}

\newcommand{\syntaxtranslation}{
  \[
    \begin{array}{llccl}
      \translate{()} &= \textbf{ret } () & & \translate{x} &= x! \\
      \translate{\lambda x.e} &= \lambda x. \translate{e} & &
      \translate{e_1\;e_2} &= \translate{e_1} \{\translate{e_2}\}\\
      \translate{\textbf{let } x = e_1 \textbf{ in } e_2} &=
      x \leftarrow \textbf{ret }\set{\translate{e_1}} \textbf{ in } \translate{e_2} & &
      \translate{e_1; e_2} &= x \leftarrow \translate{e_1} \textbf{ in } x; \translate{e_2} \\
      \translate{(e_1, e_2)} &= \textbf{ret }(\{\translate{e_1}\}, \{\translate{e_2}\}) & &
      \translate{\textbf{sub } e} &= \textbf{sub } \translate{e} \\
      \translate{\textbf{inl } e } &= \textbf{ret } (\textbf{inl }\{\translate{e}\}) & &
      \translate{\textbf{inr } e } &= \textbf{ret } (\textbf{inr }\{\translate{e}\}) \\
    \end{array}
  \]
  \vspace*{-0.75\baselineskip}
  \begin{align*}
    \translate{\textbf{let } (x_1, x_2) =  e_1 \texttt{ in } e_2} &=
    x \leftarrow \translate{e_1} \textbf{ in } (x_1, x_2) \leftarrow x \textbf{ in } \translate{e_2} \\
    \translate{\textbf{case} \texttt{ } e_1 \textbf{ of inl } x_1 \rightarrow e_2 \textbf{ inr } x_2 \rightarrow e_3} &=
    x \leftarrow \translate{e_1} \textbf{ in } \textbf{case} \texttt{ } x \textbf{ of inl } x_1 \rightarrow \translate{e_2} \textbf{ inr } x_2 \rightarrow \translate{e_3}\\
  \end{align*}
}

\begin{figure}
  \begin{align*}
    \translate{\texttt{unit}} &= (\textbf{F } \texttt{unit}, \coeffcolor{\coeffoverline{\qlazy}}) \\
    \translate{ (x :^{\coeffsymbol} \tau_1^{\coeffcolor{\gamma_1}}) \xrightarrow{\coeffcolor{\gamma_2}} \tau_2} &=
    ((\utype{\gamma_1 + \gamma_1'} B_1)^{\coeffcolor{\coeffsymbol + \gamma_2'(x)}} \rightarrow  \downshift{x} B_2, \coeffcolor{\gamma_2 + \downshift{x} \gamma_2'}) \textit{ where }
    \translate{\tau_1} = (B_1, \coeffcolor{\gamma_1'}) \textit{ and } \translate{\tau_2} = (B_2, \coeffcolor{\gamma_2'}) \\
    \translate{\tau_1^{\coeffcolor{\gamma_1}} \times \tau_2^{\coeffcolor{\gamma_2}}} &= (\textbf{F }
    (\utype{\gamma_1 + \gamma_1'} B_1 \times \utype{\gamma_2 + \gamma_2'} B_2), \coeffcolor{\coeffoverline{\qlazy}})  \textit{ where }
    \translate{\tau_1} = (B_1, \coeffcolor{\gamma_1'}) \textit{ and } \translate{\tau_2} = (B_2, \coeffcolor{\gamma_2'}) \\
    \translate{\tau_1^{\coeffcolor{\gamma_1}} + \tau_2^{\coeffcolor{\gamma_2}}} &= (\textbf{F }
    (\utype{\gamma_1 + \gamma_1'} B_1 + \utype{\gamma_2 + \gamma_2'} B_2), \coeffcolor{\coeffoverline{\qlazy}})  \textit{ where }
    \translate{\tau_1} = (B_1, \coeffcolor{\gamma_1'}) \textit{ and } \translate{\tau_2} = (B_2, \coeffcolor{\gamma_2'})
  \end{align*}
  \vspace*{-0.75\baselineskip}
  \begin{align*}
    \translate{\cdot} = \cdot \;\;\;\;\;\;\;\;\;\;\;\;\;\;\;\;
    \translate{\Gamma, x:^{\coeffcolor{\gamma}} \tau} =
    \translate{\Gamma}, x : \utype{\gamma + \gamma'} B \textit{ where }
    \translate{\tau} = (B, \coeffcolor{\gamma'})
  \end{align*}

  \syntaxtranslation

  \ifextended \vspace*{-1.5\baselineskip} \else \vspace*{-1.5\baselineskip} \fi

  \caption{Translation from \cbnlangname to \cbpvlangname}

  \ifextended \vspace*{-0.75\baselineskip} \else \vspace*{-.5\baselineskip} \fi
  \label{fig:cbn-trans}
\end{figure}

The \cbnlangname typing rules are related those of \cbpvlangname by
the translation in Figure \ref{fig:cbn-trans},
which is identical to Levy's save for the inclusion of $\coeffcolor{\gamma}$s
and some additional bookkeeping in the type translation.
Because of the difference in strictness tracking behavior between arrow types in \cbnlangname and \cbpvlangname,
we need to track the attributes suspended by the \cbnlangname arrow type separately,
so we place them on any outer $\utype{}$ that may occur in the type translation.
Therefore, the type translation is a function $\tau \rightarrow B \times \coeffcolor{\gamma}$,
where the $\coeffcolor{\gamma}$ tracks any effects suspended in the \cbnlangname types that become
unsuspended vectors in the resulting \cbpvlangname derivation.

We now state and prove our main lemma:

\begin{lemma}[\cbnlangname Translation Correctness]
  \label{lem:translation-sound}
  If $\Gamma \cbnvdash e :^{\coeffcolor{\gamma}} \tau$ and
  $\translate{\tau} = (B, \coeffcolor{\gamma'})$, then
  $\coeffcolor{\gamma + \gamma'} \cdot \translate{\Gamma} \vdash \translate{e} : B$.
\end{lemma}

\begin{proof}
  By induction on the \cbnlangname typing derivation.
\end{proof}

The difference in the behavior of functions between \cbnlangname and \cbpvlangname
causes a slight discrepancy in the translation:
as seen in the statement of Lemma \ref{lem:translation-sound},
the $\coeffcolor{\gamma'}$ from the translation of $\tau$ must be added to the $\coeffcolor{\gamma}$
in the \cbnlangname typing derivation to yield the strictness computed by
\cbpvlangname.
However, inspection of the translation for types in Figure \ref{fig:cbn-trans}
shows that, at returner types, $\coeffcolor{\gamma'}$ is always $\lazyvec$, the unit of our effect algebra.
We thus state our main theorem at $\ftype$ types.

\begin{theorem}[\cbnlangname Translation Preserves Strictness]
  \label{thm:translation-sound}
  If $\Gamma \cbnvdash e :^{\coeffcolor{\gamma}} \tau$ and
  $\translate{\tau} = (\ftype\;A, \coeffcolor{\gamma'})$, then
  $\coeffcolor{\gamma} \cdot \translate{\Gamma} \vdash \translate{e} : \ftype\; A$.
\end{theorem}

\begin{proof}
  Using Lemma \ref{lem:translation-sound} and case analysis on $\tau$.
\end{proof}

We can also show that the semantics of \cbpvlangname ``interpret'' \cbnlangname in a $\coeffcolor{\gamma}$-preserving way.

\begin{theorem}[\cbnlangname Interpretation]
  \label{thm:interpretation}
  If $\Gamma \cbnvdash e :^{\coeffcolor{\gamma}} \tau$ and $\translate{\tau} = (\ftype\;A, \coeffcolor{\gamma'})$,
  then, for any $\rho$ such that $\rho \vDash \translate{\Gamma}$, there is some $W$ such that
  $\coeffcolor{\gamma} \cdot \rho \vdash \translate{e} \Downarrow \textbf{ret }W$.
\end{theorem} %

\begin{proof}
  Simple composition of Lemma \ref{lem:fundamental-lemma} and Theorem \ref{thm:translation-sound}.
\end{proof}

With this, we have a strong claim that the type system we have developed
for \cbpvlangname can model the \cbnlangname type-and-effect system
we laid out in Section \ref{sec:cbn}.

\subsection{Aside: Are Attribute Vectors Effects?} \label{subsec:naming}

One unusual facet of the relationship between \cbnlangname and \cbpvlangname
is that $\coeffcolor{\gamma}$s behave like effects in \cbnlangname
but lose this flavor when translated to \cbpvlangname.
This seems strange on its face, and it is worth discussing why we describe
$\coeffcolor{\gamma}$s as effects in \cbnlangname but
not in \cbpvlangname.

\subsubsection*{Effects in \cbnlangname}

CBN languages generally prefer to isolate effects via monads to maintain purity,
so effect systems for such languages are relatively uncommon.
\cbnlangname does, however, resemble the call-by-name effect system \textsc{Name} described by
\citet{McDermottMycroft-2018}.
In particular, while \textsc{Name} is a general effect system,
if one were to instantiate that system with the $\coeffcolor{\gamma}$ effects
we have described, the result would be strikingly similar to \cbnlangname.
Much like \cbnlangname, \textsc{Name} has typing contexts that associate variables with
both types and effects, produces latent effects in its variable rule,
and has effect annotations on both function parameters and bodies.
The main difference between the systems is that the $\rref{T-CBN-Var}$ rule
marks its variable as strict in the effect it
produces---necessary for the specific effect that \cbnlangname is tracking.
% \bcp{That's all pretty clear except for the last half sentence.}
% \dhs{Probably too much detail. Does this clear things up? }
% BCP: Yep

Additionally, as noted above, \cbnlangname's $\coeffcolor{\gamma}$s
form a preordered monoid,
which, according to \citet{katsumata-2014} is necessary for them to be effects.
The preorder is required to compare effects, while the monoid operator
is necessary to combine them in function applications.

\subsubsection*{Coeffects in \cbpvlangname?}

The rules of \cbpvlangname, on the other hand,
mostly do not resemble effect rules.
The most obvious reason for this is that the typing judgment for values
contains $\coeffcolor{\gamma}$s;
recalling Levy's slogan for CBPV---``a value is, a computation
does''---it would be quite strange for a value to produce an effect.
On the other hand,
% while coeffects in CBPV are not yet fully understood,
the rules for \cbpvlangname do resemble the coeffect rules in prior work by \citet{torczon-effects-2024},
except that \cbpvlangname tracks $\coeffcolor{\gamma}$s on the $\utype{}$ type rather than the $\ftype$.

Evidence for this connection can be found in the way our type system degenerates when
instantiated for a call-by-value language without thunks.
Since \cbpvlangname can also model CBV evaluation \cite{levy-call-by-push-value-1999},
one might wonder what the pre-image of the translation of CBV into \cbpvlangname would look like.
However, since CBV does not suspend computation, a type system tracking lazy usage devolves into a relevance type system,
which is known to be an instance of a coeffect system \cite{abel-coeffects,Petricek-coeffect-2014, brady-idris-2021}.

Another technical difference between general graded type systems and
ours relates to the structure of the associated modality.
In CBPV, the graded Box and Diamond modalities can be derived by composing the $\utype{}$ and $\ftype$ types.
Prior work \cite{torczon-effects-2024} grades coeffects using the $\ftype$ type and effects using the $\utype{}$ type.
In contrast, our system tracks variable usage like a coeffect but grades the $\utype{}$ type like an effect.
It could be, however, that there is flexibility where grades appear
(e.g. other prior work \cite{rioux-computation-2020} grades the $\ftype$ type with effects),
so the relationship is not completely clear.

Lastly, coeffect algebras feature a multiplication operator to scale coeffects,
along with an addition operator to combine them, but
\cbpvlangname lacks such a scaling operation.
This is because the strictness algebra of \cbpvlangname is
\textit{non-quantitative} \cite{Bianchini-resource-2023} (or \textit{informational} \cite{abel-unified-2020})
and defines $\coeffcolor{+}$ idempotently. Multiplication becomes degenerate in such a system, since it
is just iterated addition.
To make \cbpvlangname's rules more closely resemble coeffect rules,
we could define a scaling operation $\coeffcolor{\cdot}$ with $\qstrict$ as its unit and
$\qlazy$ as its annihilator. However,
including explicit scaling would not change any attributes computed
by \cbpvlangname, so
there does not seem to be much point.

\begin{comment}
The key difference between Torczon et al.'s coeffect rules and those of
\cbpvlangname is the treatment of the $\utype{}$s,
where \cbpvlangname more closely resembles their effect system.
Their system captures coeffects on $\ftype$s instead,
but it is not necessarily the case that coeffects must be placed on
$\ftype$ types, nor that whatever is captured on $\ftype$ types necessarily must be coeffects.
In particular, other prior work on effects in CPBV
has placed effects on both $\ftype$ types \cite{rioux-computation-2020, mcdermott-grading-2025}
and on $\utype{}$ types \cite{torczon-effects-2024}.
If effects can be modeled equally effectively on $\ftype$s and $\utype{}$s,
there is little reason to think that the same would not be true of coeffects.
Further investigation of this difference is a promising avenue for future work.
\end{comment}

\section{Semantic Properties of \cbpvlangname} \label{sec:bot}

Our next job is to show that the intensional, type-theoretic
characterization of strictness
provided by \cbpvlangname refines the commonly understood extensional definition.
To do this, we allow programs to be evaluated in environments that
may not have bindings for every variable in the typing context.
Such missing variables cannot be scrutinized in any valid semantic derivation tree
(since the premise of the $\rref{E-Var}$ rule cannot be satisfied),
so programs cannot use an unbound value strictly.
(We deliberately remain agnostic as to how such variables may end up missing
  from the environment, so a variable
without a binding can be used to model one given a $\bot$ value in the extensional model.)
% \bcp{??  }

We establish two different properties: \textit{lazy
soundness} and \textit{strict failure}.
First, any well-typed program can be run in an environment lacking bindings
for any variables with an $\qlazy$ attribute and still produce a result.
Second, a well-typed program with an $\ftype$ type is guaranteed to fail when run in an environment lacking a value
for any variable with an $\qstrict$ attribute.

\newcommand{\scope}{X}

\newcommand{\observer}{\coeffcolor{\gamma_\texttt{Obs}}}
\newcommand{\vDashlazy}{\vDash^{\qlazy}}
\newcommand{\vDashstrict}{\vDash^{\qstrict}}
\newcommand{\vDashunknown}{\vDash^{\qunknown}}
\newcommand{\fail}{\coeffcolor{\lightning}}
\newcommand{\vDashlighting}{\vDash^{\fail}}

\subsection{Lazy Soundness}

\begin{figure}
  \label{fig:lrb}
  \begin{align*}
    \lrv{\texttt{unit}}_\scope^{\observer} &\;=\; \set{()} \\
    \lrv{A_1 \times A_2}_\scope^{\observer} &\;=\; \set{(W_1, W_2) \alt W_1 \in \lrv{A_1}_\scope^{\observer} \andtext W_2 \in \lrv{A_2}_\scope^{\observer}} \\
    \lrv{A_1 + A_2}_\scope^{\observer} &\;=\; \set{\textbf{inl } W_1 \alt W_1 \in \lrv{A_1}_\scope^{\observer}} \cup \set{\textbf{inr } W_2 \alt W_2 \in \lrv{A_2}_\scope^{\observer}} \\
    \lrv{\textbf{U}_{\coeffcolor{\gamma}} B}_\scope^{\observer} &\;=\; \set{\{{\coeffcolor{\gamma'}}, \rho, M \} \alt
      \coeffcolor{\gamma \mid_{\dom{\gamma'}} = \gamma' \mid_{\dom{\gamma}}} \andtext \\
      &\;\;\;\;\;\;\;\;\;\;(\coeffcolor{\observer} \mid_{\scope} \;\;\coeffcolor{\leq_S}\; \coeffcolor{\gamma} \mid_{\scope} \impliestext
    {\coeffcolor{\gamma'}} \cdot \rho \vdash M \Downarrow T \andtext T \in \lrc{B}_\scope^{\observer})} \\
    \lrc{\textbf{F} A }_\scope^{\observer} &\;=\; \set{\textbf{ret } W \alt W \in \lrv{A}_\scope^{\observer}} \\
    \lrc{A^{\coeffsymbol} \rightarrow B}_\scope^{\observer} &\;=\; \set{\llangle \coeffcolor{\gamma}, \rho, \lambda x. M \rrangle \alt
      \foralltext W \in \lrv{A}_\scope^{\observer},
    \coeffcolor{\gamma}\cdot \rho, x \mapsto^{\coeffsymbol} W\ \vdash M \Downarrow T \andtext T \in \lrc{B}_\scope^{\observer}}
  \end{align*}
  \vspace*{-1\baselineskip}
  \caption{Lazy Logical Relation}
  \ifextended \vspace*{-1\baselineskip} \else \vspace*{-1\baselineskip} \fi
  \label{fig:lrb}
\end{figure}

To prove lazy soundness---that unbound lazy variables do not prevent successful execution---we
use the logical relation in Figure \ref{fig:lrb},
which is indexed by a set of variables $\scope$ and a vector $\observer$.
The set $\scope$ represents the outermost scope of the program,
i.e., the scope of the top-level program declaration.
This is the smallest scope that will ever exist during evaluation,
since variables that are free at the top level can never go out of scope.
We define also a relation $\coeffcolor{\leq_S}$ on vectors that
compares attributes pointwise and relates two vectors if the attributes
in the lesser vector are ``stricter'' than those in the greater.
That is, anywhere the greater vector is non-lazy, the lesser must be as well.

The $\observer$ functions like an ``observer'' context,
describing how a variable can be used (as opposed to how it is actually used).
This gives rise to the requirement that $\observer$ always be stricter than
the $\coeffcolor{\gamma}$ on the type in the $\utype{}$ case of the relation.
Thunks in this relation may not evaluate to a value,
but, if they do, they will use variables less strictly than what the $\observer$ vector allows.
For example,
the thunk $\{\textbf{ret } x\} \in \lrv{\utype{x:\qstrict} \ftype A}_{\{x\}}^{\coeffcolor{x:\qlazy}}$,
even if $x$ is missing a value.
The observer vector in this case tells us that $x$ will not be used strictly
at any point, so this thunk is never forced and accordingly is still
``well typed'' by this relation
because $\coeffcolor{x:\qlazy \nleq_S x:\qstrict}$.

We define semantic typing for environments, values, and computations in
Figure \ref{fig:bot-semtyping-lazy}.
Intuitively, an environment $\rho$ is semantically well typed---by a vector $\coeffcolor{\gamma}$ and
context $\Gamma$ at a scope $X$ and observer $\observer$---if
all the values associated with variables that have attributes other than $\qlazy$
are well typed at their associated types in $\Gamma$ at the same scope and observer.

A value or computation is semantically well typed by vector $\coeffcolor{\gamma}$ and context $\Gamma$
if for every observer $\observer$, environment $\rho$, and scope $\scope$,
the value or computation can be evaluated to a well-typed result as long as $\rho$ is well typed,
$\observer$ is stricter than the vector $\coeffcolor{\gamma}$ of attributes
used by evaluation,
and every variable in $\observer$ that is not present in $\scope$
is assigned an attribute of $\qstrict$, allowing it to be used.

\begin{figure}
  \begin{align*}
    \coeffcolor{\gamma} \cdot \Gamma \vDash_\scope^{\observer} \rho &\;\;\;\triangleq\;\;\;
    x :^{\coeffsymbol} A \in (\coeffcolor{\gamma} \cdot\Gamma) \andtext \coeffcolor{\coeffsymbol\neq \qlazy} \impliestext
    x \mapsto W \in \rho \andtext W \in \lrv{A}_\scope^{\observer} \\
    \coeffcolor{\gamma}  \cdot \Gamma \vDashlazy \rho &\;\;\;\triangleq\;\;\; \coeffcolor{\gamma}  \cdot \Gamma \vDash_{\coeffcolor{\dom{\gamma}} }^{\coeffcolor{\gamma} } \rho \\
    X \vDashstrict \observer &\;\;\;\triangleq\;\;\; x \notin X \impliestext x: \qstrict \in\coeffcolor{\observer} \\
    \coeffcolor{\gamma} \cdot \Gamma \vDashlazy V : A &\;\;\;\triangleq\;\;
    \foralltext X \andtext \observer, (\observer \cdot \Gamma \vDash_\scope^{\observer} \rho \andtext
      \coeffcolor{\observer} \mid_{\scope} \;\;\coeffcolor{\leq_S}\; \coeffcolor{\gamma} \mid_{\scope} \andtext
    X \vDashstrict \observer) \impliestext \\
    &\;\;\;\;\;\;\;\;\;\;\coeffcolor{\gamma}  \cdot \rho \vdash V \Downarrow W \andtext W \in \lrv{A}_\scope^{\observer} \\
    \coeffcolor{\gamma}  \cdot \Gamma \vDashlazy M : B &\;\;\;\triangleq\;\;
    \foralltext X \andtext \observer, (\observer \cdot \Gamma \vDash_\scope^{\observer} \rho \andtext
      \coeffcolor{\observer} \mid_{\scope} \;\;\coeffcolor{\leq_S}\; \coeffcolor{\gamma} \mid_{\scope} \andtext
    X \vDashstrict \observer) \impliestext \\
    &\;\;\;\;\;\;\;\;\;\;\coeffcolor{\gamma}  \cdot \rho \vdash M \Downarrow T \andtext T \in \lrc{B}_\scope^{\observer}
  \end{align*}
  \vspace*{-1\baselineskip}
  \caption{Lazy Semantic Typing}
  \vspace*{-0.5\baselineskip}
  \label{fig:bot-semtyping-lazy}
\end{figure}

This last requirement deserves some additional explanation.
As mentioned, $\scope$ describes the outermost
scope of a program,
and in the proof of Theorem \ref{thm:soundness-bot} we will choose $\scope$ to be this scope.
For the purposes of induction, however, we must generalize to an arbitrary scope
that exactly describes the variables that might be missing values in $\rho$.
Every variable not in $\scope$ must have been added to $\rho$
during evaluation, and thus cannot be missing
and can have an $\qstrict$ attribute in $\observer$.

Now we can state and prove the fundamental lemma,
from which soundness follows directly.

\begin{lemma}[Lazy Fundamental Lemma]
  \label{lem:fundamental-lemma-bot}
  For all $\coeffcolor{\gamma}$ and $\Gamma$,
  if $\coeffcolor{\gamma}  \cdot \Gamma \vdash V : A$,
  then $\coeffcolor{\gamma}  \cdot \Gamma \vDashlazy V : A$,
  and if $\coeffcolor{\gamma}  \cdot \Gamma \vdash M : B$,
  then $\coeffcolor{\gamma}  \cdot \Gamma \vDashlazy M : B$.
\end{lemma}

\begin{proof}
  By induction on the typing derivation.
\end{proof}

\begin{theorem}[Lazy Soundness]
  \label{thm:soundness-bot}
  For all $\coeffcolor{\gamma} $, $\Gamma$, and $\rho$ such that
  $\coeffcolor{\gamma}  \cdot \Gamma \vDashlazy \rho$
  and all $x$ such that $x : \qlazy \in \coeffcolor{\gamma}$,
  \begin{enumerate}
    \item if $\coeffcolor{\gamma}  \cdot \Gamma \vdash V : A$,
      then there is some $W$ such that
      $\coeffcolor{\gamma}  \cdot (\rho-x) \vdash V \Downarrow W$;
    \item if $\coeffcolor{\gamma}  \cdot \Gamma \vdash M : B$,
      then there is some $T$ such that
      $\coeffcolor{\gamma}  \cdot (\rho-x) \vdash M \Downarrow T$,
  \end{enumerate}
  \hspace{\parindent} where $\rho-x$ denotes the environment equivalent to $\rho$ with the binding for $x$ removed.
\end{theorem}

\begin{proof}
  Using the fundamental lemma,
  choosing $\observer$ to be $\coeffcolor{\gamma}$ and $\scope$ to be $\coeffcolor{\dom{\gamma}}$.
  We also rely on the fact that $\coeffcolor{\gamma}  \cdot \Gamma \vDashlazy \rho-x$ when
  $x : \qlazy \in \coeffcolor{\gamma}$.
\end{proof}

This theorem tells us that
if \cbpvlangname assigns a lazy attribute to a variable in some program,
we can still evaluate that program in an environment where that variable is not bound to a value.

\subsection{Strict Failure}

\newcommand{\coeffequiv}{\coeffcolor{\equiv}}

\newcommand{\equivdef}{
  $$
  \infer
  {() \;\coeffequiv\; ()}{} \qquad
  \infer
  {\textbf{inl } W \;\coeffequiv\; \textbf{inl } W'}{W \;\coeffequiv\; W'} \qquad
  \infer
  {\textbf{inr } W \;\coeffequiv\; \textbf{inr } W'}{W \;\coeffequiv\; W'} \qquad
  \infer
  {(W_1, W_2) \;\coeffequiv\; (W_1', W_2')}{W_1 \;\coeffequiv\; W_1' & W_2 \;\coeffequiv\; W_2'}
  \vspace{2.75ex}
  $$
  $$
  \infer
  {\set{\coeffcolor{\gamma}, \rho, M} \;\coeffequiv\; \set{\coeffcolor{\gamma'}, \rho', M}}{\rho \;\coeffequiv\; \rho'} \qquad
  \infer
  {\textbf{ret } W \;\coeffequiv\; \textbf{ret } W'}{W \;\coeffequiv\; W'} \qquad
  \infer
  {\llangle\coeffcolor{\gamma}, \rho, M\rrangle \;\coeffequiv\; \llangle\coeffcolor{\gamma'}, \rho', M\rrangle}{\rho \;\coeffequiv\; \rho'} \qquad
  \infer
  {\rho \;\coeffequiv\; \rho'}{\forall x, \rho(x) \;\coeffequiv\; \rho'(x)}
  $$
}

\begin{figure}
  \equivdef
  \caption{Equivalence modulo $\coeffcolor{\gamma}$}
  \ifextended \vspace*{-0.5\baselineskip} \else \vspace*{-0.75\baselineskip} \fi
  \label{fig:equal-erase}
\end{figure}

To show that strictly using a missing value always causes evaluation to fail,
we need a different, more complex logical relation.
We also need a precise notion of what it means for evaluation to fail.
In particular, it is possible for a semantic derivation in \cbpvlangname
to fail because the wrong attributes were chosen for a non-deterministic
rule like \rref{E-Thunk} or \rref{E-Sub}, rather than because a value is absent.
This kind of failure is uninteresting,
in the sense that evaluation could simply have made a different choice and produced a valid derivation tree;
we want to argue that strict usage of a variable whose value is missing
cannot \emph{ever} produce a valid semantic derivation, regardless of how attributes are chosen.
We denote this kind of failure with the $\fail$ symbol.

We formalize a notion of derivations that are equivalent up to their
choice of attributes by defining a relation $\coeffequiv$ that relates
two closed terminals if they are equivalent modulo their attributes.
This definition is given in Figure \ref{fig:equal-erase}.
\ifextended

Using this definition, we can observe a fact about \cbpvlangname:
while its semantics is not deterministic in general (in particular, whenever the \rref{E-Sub} rule is used),
it \textit{is} deterministic modulo $\coeffcolor{\gamma}$s. Formally:

\begin{lemma}[Determinism]   \label{lem:deterministic}
  Given $\rho_1$ and $\rho_2$ such that $\rho_1 \coeffequiv\rho_2$,
  \begin{enumerate}
    \item If $\coeffcolor{\gamma_1} \cdot\rho_1 \vdash V \Downarrow W_1$ and
      $\coeffcolor{\gamma_2} \cdot\rho_2 \vdash V \Downarrow W_2$ then $W_1\coeffequiv W_2$;
    \item If $\coeffcolor{\gamma_1} \cdot\rho_1 \vdash M \Downarrow T_1$ and
      $\coeffcolor{\gamma_2} \cdot\rho_2 \vdash M \Downarrow T_2$ then $T_1\coeffequiv T_2$
  \end{enumerate}
\end{lemma}

\else \fi
We then use this notion of equivalence modulo attributes to precisely express what it means
for there to be no possible valid derivation tree for a given expression and
environment.

\begin{definition}[Semantic Failure]   \label{def:failure}
  \begin{align*}
    \rho \vdash V \fail &\;\;\;\triangleq\;\;\; \foralltext \coeffcolor{\gamma} \andtext \rho',\textit{\;\;if\;\;}
    \rho \;\coeffequiv\; \rho' \textit{\;\;then\;\;} \nexists W, \coeffcolor{\gamma} \cdot \rho' \vdash V \Downarrow W \\
    \rho \vdash M \fail &\;\;\;\triangleq\;\;\; \foralltext \coeffcolor{\gamma} \andtext \rho', \textit{\;\;if\;\;}
    \rho \;\coeffequiv\; \rho' \textit{\;\;then\;\;} \nexists T, \coeffcolor{\gamma} \cdot \rho' \vdash M \Downarrow T
  \end{align*}
\end{definition}
Intuitively, this definition says that a value or computation
fails to evaluate in an environment $\rho$ if, for any choice of
attributes---both those in the derivation itself and also those appearing in the
environment---there is no closed terminal that can be produced by any derivation.
This rules out trivial failures of evaluation that occur due to an incorrect choice of attributes
and instead tells us about terms that fail to evaluate because they attempt to use
a missing or ill-typed value.

\begin{figure}
  \begin{align*}
    \lrv{\texttt{unit}}_x &= \set{()} \\
    \lrv{A_1 \times A_2}_x &= \set{(W_1, W_2) \alt W_1 \in \lrv{A_1}_x \andtext W_2 \in \lrv{A_2}_x} \\
    \lrv{A_1 + A_2}_x &= \set{\textbf{inl } W_1 \alt W_1 \in \lrv{A_1}_x} \cup \set{\textbf{inr } W_2 \alt W_2 \in \lrv{A_2}_x} \\
    \lrv{\textbf{U}_{\coeffcolor{\gamma}} B}_x &= \set{\{{\coeffcolor{\gamma'}}, \rho, M \} \alt
      \coeffcolor{\gamma \mid_{\dom{\gamma'}} = \gamma' \mid_{\dom{\gamma}}} \andtext
    (\coeffcolor{\gamma'}, \rho, M) \in\lrm{B}_x} \\
    \lrc{\textbf{F} A }_x &= \set{\textbf{ret } W \alt W \in \lrv{A}_x} \\
    \lrc{A^{\coeffsymbol} \rightarrow B}_x &= \set{\llangle \coeffcolor{\gamma}, \rho, \lambda x. M \rrangle \alt
      \foralltext W \in \lrv{A}_x, \\
      &\;\;\;\;\;\;\;\;\rho, x \mapsto W \vdash M \fail \ortext
    \coeffcolor{\gamma}\cdot \rho, x \mapsto^{\coeffsymbol} W\ \vdash M \Downarrow T \andtext T \in \lrc{B}_x} \\
  \end{align*}
  \vspace*{-1.5\baselineskip}
  \begin{align*}
    \lrm{\ftype A}_x &= \set{(\coeffcolor{\gamma}, \rho, M) \alt
      \coeffcolor{x : \qstrict \in\gamma} \impliestext \rho \vdash M \fail \andtext \\
      &\;\;\;\;\;\;\;\;\;\;\;\;\;\;\;\;\;\;\;\;\;\;\;\coeffcolor{x : \qstrict \notin\gamma} \impliestext
      \rho \vdash M \fail \ortext
    \coeffcolor{\gamma}\cdot \rho \vdash M \Downarrow T \andtext T \in \lrc{\ftype A}_x} \\
    \lrm{A^{\coeffsymbol} \rightarrow B}_x &= \set{(\coeffcolor{\gamma}, \rho,  M) \alt
      \coeffcolor{x : \qstrict \in\gamma} \impliestext
      \rho \vdash M \fail \ortext
      (\coeffcolor{\gamma}\cdot \rho \vdash M \Downarrow T \andtext T \in \lrf{\coeffsymbol, A, B}_x) \andtext \\
      &\;\;\;\;\;\;\;\;\;\;\;\;\;\;\;\;\;\;\;\;\;\;\;\coeffcolor{x : \qstrict \notin\gamma} \impliestext
      \rho \vdash M \fail \ortext
    \coeffcolor{\gamma}\cdot \rho \vdash M \Downarrow T \andtext T \in \lrc{A^{\coeffsymbol} \rightarrow B}_x} \\
  \end{align*}
  \vspace*{-1.5\baselineskip}
  \begin{align*}
    \lrf{(\coeffsymbol, A, \ftype A')}_x &= \set{\llangle \coeffcolor{\gamma}, \rho, \lambda x. M \rrangle \alt
    \foralltext W \in \lrv{A}_x, \;\;\;\;\rho, x \mapsto W\ \vdash M \fail} \\
    \lrf{(\coeffsymbol, A, A'^{\coeffcolor{\coeffsymbol'}} \rightarrow B)}_x &= \set{\llangle \coeffcolor{\gamma}, \rho, \lambda x. M \rrangle \alt
      \foralltext W \in \lrv{A}_x, \rho, x \mapsto W\ \vdash M \fail \ortext \\
      &\;\;\;\;\;\;\;\;
    \coeffcolor{\gamma}\cdot \rho, x \mapsto^{\coeffsymbol} W\ \vdash M \Downarrow T \andtext T \in \lrf{(\coeffcolor{\coeffsymbol'}, A', B)}_x}
  \end{align*}
  \vspace*{-\baselineskip}
  \caption{Strict Logical Relation}
  \vspace*{-\baselineskip}
  \label{fig:lrstrict}
\end{figure}

We define the logical relation for the proof of strict failure
in Figure \ref{fig:lrstrict}.
Unlike the lazy relation, the strict relation is indexed not by an entire
attribute vector but rather by a single variable.
We can understand this logical relation as searching the environment for the site of execution failure:
the $x$ variable indexing the relation is the location of the variable missing its value,
and the relation guarantees that thunk values that are in it
will necessarily fail when forced, if their attributes use $x$ strictly.
If a thunk's attributes are not $\qstrict$ for $x$, the relation
instead guarantees that forcing the thunk will either fail
or produce a result that itself is in the relation.

Most of the complexity in this relation comes from the fact that
attributes ``pass through'' lambda abstractions in \cbpvlangname.
This has precedent elsewhere in the CBPV literature \cite{torczon-effects-2024, kammar-2013,kammar-2012},
but it comes with an unfortunate consequence:
despite the fact that, say, $\lambda x. \;\textbf{ret } y$
typechecks with attribute $\coeffcolor{y:\qstrict}$, evaluation of this program will not actually fail.
Instead, it will produce a closure.
In practice, to bind a function to a variable in \cbpvlangname,
it must be suspended via a thunk, which will restore the behavior we want.
Therefore, we take the standard approach in the CBPV literature,
stating and proving our top-level soundness theorem at returner types $\ftype A$ only.
However, for purposes of induction the actual relation itself
must be fully general.

\begin{figure}
  \begin{align*}
    \coeffcolor{\gamma} \cdot \Gamma\vDashlighting_x \rho &\;\;\;\triangleq\;\;\; \coeffcolor{x : \qstrict \in\gamma} \andtext x \notin\rho \andtext
    (y \neq x \impliestext \rho(y) \in\lrv{\Gamma(y)}_x) \\
    \coeffcolor{\gamma} \cdot \Gamma\vDashunknown_x \rho &\;\;\;\triangleq\;\;\; \coeffcolor{x : \qstrict \notin\gamma} \andtext x \notin\rho \andtext
    (y \neq x \impliestext \rho(y) \in\lrv{\Gamma(y)}_x) \\
    \coeffcolor{\gamma} \cdot \Gamma \vDashstrict V : A &\;\;\;\triangleq\;\;\;
    \coeffcolor{\gamma} \cdot \Gamma\vDashlighting_x \rho \impliestext
    \rho \vdash V \fail \andtext \\
    &\;\;\;\;\;\;\;\;\;\;\coeffcolor{\gamma} \cdot \Gamma\vDashunknown_x \rho \impliestext
    \rho \vdash V \fail \ortext
    \coeffcolor{\gamma}\cdot \rho \vdash V \Downarrow W \andtext W \in\lrv{A}_x \\
    \coeffcolor{\gamma} \cdot \Gamma \vDashstrict M : \ftype A &\;\;\;\triangleq\;\;\;
    \coeffcolor{\gamma} \cdot \Gamma\vDashlighting_x \rho \impliestext
    \rho \vdash M \fail \andtext \\
    &\;\;\;\;\;\;\;\;\;\;\coeffcolor{\gamma} \cdot \Gamma\vDashunknown_x \rho \impliestext
    \rho \vdash M \fail \ortext
    \coeffcolor{\gamma}\cdot \rho \vdash M \Downarrow T \andtext T \in\lrc{\ftype A}_x \\
    \coeffcolor{\gamma} \cdot \Gamma \vDashstrict M : A^{\coeffsymbol} \rightarrow B &\;\;\;\triangleq\;\;\;
    \coeffcolor{\gamma} \cdot \Gamma\vDashlighting_x \rho \impliestext
    \rho \vdash M \fail \ortext
    \coeffcolor{\gamma}\cdot \rho \vdash M \Downarrow T \andtext T \in\lrf{A^{\coeffsymbol} \rightarrow B}_x \andtext \\
    &\;\;\;\;\;\;\;\;\;\;\coeffcolor{\gamma} \cdot \Gamma\vDashunknown_x \rho \impliestext
    \rho \vdash M \fail \ortext
    \coeffcolor{\gamma}\cdot \rho \vdash M \Downarrow T \andtext T \in\lrc{A^{\coeffsymbol} \rightarrow B}_x
  \end{align*}
  \vspace*{-\baselineskip}
  \caption{Strict Semantic Typing}
  \ifextended \vspace*{-\baselineskip} \else \fi
  \label{fig:bot-semtyping-strict}
\end{figure}

To achieve this generality, we define two different helper relations
$\lrm{\cdot}_x$ (for \textit{maybe}) and $\lrf{\cdot, \cdot, \cdot}_x$ (for \textit{failure}).
The former describes computations that \emph{may} fail depending on their usage of $x$,
while the latter describes closures that \emph{definitely} fail due to their usage of $x$.
On function types returning $\ftype$ types,
this latter relation guarantees that evaluation of the closed-over function body actually fails,
while on function types returning other functions it instead guarantees that evaluation of the
function body will either fail or produce another closure that is also in
$\lrf{\cdot, \cdot, \cdot}_x$.

Figure \ref{fig:bot-semtyping-strict} defines strict semantic typing and formalizes the intuition for the logical relation.
An environment can either be $\vDashlighting_x$ (well typed, but missing a strictly-used value at $x$) or
$\vDashunknown_x$ (well typed, but missing a possibly-used value at $x$)
for a given $\coeffcolor{\gamma}$ and $\Gamma$.
A semantically well-typed value or computation will definitely fail
(or be in the failure relation, in the case of functions) if run with an
environment that is missing a strictly-used value at $x$,
and it will either fail or produce a well-typed result if evaluated
with an environment that is missing a possibly-used value at $x$.

\ifextended
The search-like character of this relation comes from the observation that
any $\rho$ such that $\coeffcolor{\gamma_1 + \gamma_2} \cdot \Gamma\vDashlighting_x \rho$
for some attribute vectors $\coeffcolor{\gamma_1}$ and $\coeffcolor{\gamma_2}$
is only known to be $\vDashlighting_x$ at one of them. Formally:

\begin{lemma}[Decomposition of $\vDashlighting$ and $\vDashunknown$]
  \label{lem:minesweeper} Given attribute vectors $\coeffcolor{\gamma_1}$ and $\coeffcolor{\gamma_2}$,
  context $\Gamma$, and environment $\rho$:
  \begin{enumerate}
    \item If $\coeffcolor{\gamma_1 + \gamma_2} \cdot \Gamma\vDashlighting_x \rho$ then either
      $\coeffcolor{\gamma_1} \cdot \Gamma\vDashlighting_x \rho$, or
      $\coeffcolor{\gamma_1} \cdot \Gamma\vDashunknown_x \rho$ and $\coeffcolor{\gamma_2} \cdot \Gamma\vDashlighting_x \rho$,
    \item If $\coeffcolor{\gamma_1 + \gamma_2} \cdot \Gamma\vDashunknown_x \rho$ then
      $\coeffcolor{\gamma_1} \cdot \Gamma\vDashunknown_x \rho$ and
      $\coeffcolor{\gamma_2} \cdot \Gamma\vDashunknown_x \rho$.
  \end{enumerate}
\end{lemma}

\begin{proof}
  By case analysis on $\coeffcolor{\gamma_1(x)}$ and $\coeffcolor{\gamma_2(x)}$.
\end{proof}

Using this lemma, we can decompose typing rules
that are checked using the sum of two attribute vectors.
For example, the $\rref{T-Let}$ rule says
that $\coeffcolor{\gamma_1 + \gamma_2}$ can be used to check
$x \leftarrow M_1 \textbf{ in } M_2$ if
$\coeffcolor{\gamma_1}$ can check $M_1$ and
$\coeffcolor{\gamma_2}$ can check $M_2$.
If we assume that $\coeffcolor{\gamma_1 + \gamma_2}$ describes a strict use of a missing value,
Lemma \ref{lem:minesweeper} tells us that one of the two subexpressions of this let binding
will be responsible for this strict usage.
In this way we can search the evaluation derivation for the location of the actual
strict usage that fails,
and the resulting failure will cascade all the way through the tree.
\fi

As in the lazy case, the fundamental lemma derives the main theorem as a corollary.

\begin{lemma}[Strict Fundamental Lemma]
  \label{lem:fundamental-lemma-strict}
  For all $\coeffcolor{\gamma} $ and $\Gamma$,
  if $\coeffcolor{\gamma}  \cdot \Gamma \vdash V : A$
  then $\coeffcolor{\gamma}  \cdot \Gamma \vDashstrict V : A$,
  and if $\coeffcolor{\gamma}  \cdot \Gamma \vdash M : B$
  then $\coeffcolor{\gamma}  \cdot \Gamma \vDashstrict M : B$.
\end{lemma}

\begin{proof}
  \ifextended
  By induction on the typing derivation,
  making use of Lemma \ref{lem:deterministic} and Lemma \ref{lem:minesweeper}.
  \else
  By induction on the typing derivation.
  \fi
\end{proof}

\begin{theorem}[Strict Failure]
  \label{thm:bot-explosion}
  For all $x, \coeffcolor{\gamma} $, $\Gamma$, and $\rho$ such that $\coeffcolor{\gamma} \vDashlighting_x \rho$,
  \begin{enumerate}
    \item if $\coeffcolor{\gamma}  \cdot \Gamma \vdash V : A$ then
      $\rho \vdash V \fail$;
    \item if $\coeffcolor{\gamma}  \cdot \Gamma \vdash M : \ftype \text{ } A$ then
      $\rho \vdash M \fail$.
  \end{enumerate}
\end{theorem}

\begin{proof}
  Follows directly from the Fundamental Lemma \ref{lem:fundamental-lemma-strict}.
\end{proof}

This gives us the opposite guarantee from Theorem \ref{thm:soundness-bot}, telling us
that any well-typed program will fail when run in an environment where one of its strictly-used
variables is missing a value.
We can also use this theorem to argue that arguments to strict functions can be evaluated eagerly.
Informally, if some strict function is called with a thunk argument
and produces a result, Theorem \ref{thm:bot-explosion}
tells us by contraposition that the function necessarily uses the value of its argument---the value cannot be missing.
Hence, forcing the thunked argument must produce a result, and
forcing this argument before applying the function would
correspond to a strict function call.

\subsection{A More Precise Characterization of Strictness}

Having proved that \cbpvlangname's type-theoretic characterization
of intensional strictness also captures the traditional, extensional notion,
we can move on to consider the {\em new} insight the intensional
definition affords us:
the ability to distinguish between functions
that fail when ``called on bottom'' because their argument is used strictly and those
that fail independently of their argument.

The traditional abstract interpretation \cite{mycroft-theory-1980, wadler-strictness-1987}
and projection-based definitions \cite{wadler-projections-1987} of extensional strictness
assert that a function $f$ is strict if $f \bot=\bot$,
where $\bot$ denotes the result of a non-terminating or aborting computation.
Usually no distinction is made between the possible ways $\bot$ can be produced;
the definition says only that
a strict function fails whenever its argument fails.

The properties described by Theorem \ref{thm:soundness-bot} and Theorem \ref{thm:bot-explosion}
are clearly related to this definition.
As no $\bot$ value can ever actually exist,
we can represent it with an environment that is missing a binding for the variable
to which the traditional approach would ascribe a $\bot$ value.
Theorem \ref{thm:bot-explosion} then gives us a definition of strictness that implies $f \bot=\bot$:
a strict function whose argument is missing its value will fail to evaluate.
Likewise, Theorem \ref{thm:soundness-bot} gives us a definition of
laziness that implies that lazy functions do not scrutinize their arguments:
a lazy function that is given an argument whose value is missing will still evaluate.

However, intensional strictness as described by
\cbpvlangname's type system has some extra precision compared to extensional strictness.
Consider a hypothetical program
$
y \leftarrow \textbf{ret } \{\dots \textit{error} \dots\} \textbf{ in }
\lambda x. y!
$.
The function produced by this program---call it $f$---would be
considered strict, extensionally,
even though it does not use its argument at all.
This characterization is useful to an optimizing compiler performing strictness
analysis (i.e., it would allow such a compiler to evaluate this function's argument eagerly),
but it fails to describe how $f$ actually uses its argument.

Unlike the extensional definition of strictness,
\cbpvlangname does distinguish between functions that fail due to strict use of their arguments and
those that fail for other reasons.
\cbpvlangname would describe $f$ as being lazy in
its argument, rather than strict, and also tell us that $f$ uses $y$ strictly.
In a situation where a call to $f$ fails,
this additional nuance allows us to reason about \textit{why} it has failed:
\cbpvlangname tells us that this failure cannot be the result of
the function's use of its argument (due to Theorem \ref{thm:soundness-bot})
and instead tells us that it must be the result of its use of $y$ (due to Theorem \ref{thm:bot-explosion}).
This is useful information for programmers to have;
if, say, a call to a lazy function were to throw an exception,
a programmer could be certain that the source of
that exception was in the function body,
not in the argument to the call.

This increased precision regarding sources of failure also
makes intensional strictness better suited to reasoning about
impure languages where exceptions are not encapsulated in monads as in Haskell.
In such settings, the extensional definition fails to provide a satisfying
characterization of strictness because it conflates
aborting and non-terminating programs;
it would thus falsely identify a call to $f$
as an opportunity for eager evaluation, when such an optimization
may in fact change observable behavior if the argument does not terminate.
% \bcp{Why does this make its characterization of strictness bad?}

We believe
we can also apply these insights to the call-by-need
\cite{wadsworth-semantics-1971, launchbury-natural-1993} evaluation strategy used by languages like Haskell,
as call-by-need is equivalent to call-by-name \cite{mycroft-theory-1980}:
the former is a more efficient implementation of the latter that
uses values with the same strictness.
We discuss how we might formalize this connection in Section \ref{sec:future}.

\section{Tracking Unused Variables} \label{sec:absence}

\newcommand{\coeffsof}[1]{\coeffcolor{\mathbb{E}(#1)}}
\newcommand{\vdashwf}{\vdash_{\text{WF}}}
\newcommand{\qabsent}{\coeffcolor{\mathrm{U}}}
\newcommand{\lazify}[1]{\coeffcolor{\mathbb{L}(#1)}}
\newcommand{\absentvec}{\coeffoverline{\qabsent}}

\begin{wrapfigure}{r}{0.28\textwidth}
  \centering
  \ifextended
  \vspace*{-1\baselineskip}
  \else
  \vspace*{-2\baselineskip}
  \fi
  \begin{tikzpicture}

    \node at (0,0)    (U)  {$\qunknown$};
    \node at (0.75,0.75)    (S)  {$\qstrict$};
    \node at (-.75,0.75)   (L)  {$\qlazy$};
    \node at (-.75,1.5)   (A)  {$\qabsent$};

    \draw (U)   -- (L);
    \draw (A)   -- (L);
    \draw (U)   -- (S);
  \end{tikzpicture}
  \vspace*{-0.5\baselineskip}
  \caption{Extended semilattice}
  \vspace*{-0.5\baselineskip}
  \label{fig:lattice-absence}
\end{wrapfigure}

The $\qlazy$ attribute described so far is not as precise as it could be.
Unlike the $\qstrict$ attribute, which guarantees that variables
are definitely used strictly,
the $\qlazy$ attribute describes both variables
that are used lazily and those that are not used at all.
It would be more accurate to describe the $\qlazy$ attribute
as asserting that a variable is definitely \textit{not} used strictly, as opposed to
definitely used lazily.

In this section, we show how to add additional precision to the system by
extending \cbpvlangname and \cbnlangname
with a new attribute $\qabsent$ to track variables that are known to be \textit{unused}.
This fourth attribute is related to the notion of \textit{absence}
described by \textit{absence analysis} \cite{sergey-theory-2014}.
It extends the existing set to produce a new semilattice,
depicted in Figure \ref{fig:lattice-absence}.
The $\qlazy$, $\qstrict$, and $\qunknown$ attributes have the same meaning as before,
and the intuition for $\coeffcolor{\leq}$ remains the same as well:
$\qabsent$ tells us strictly more information about a variable's usage
(or lack thereof) than $\qlazy$.

\subsection{Extending \cbpvlangname}

We present the extension of \cbpvlangname first, as it is simple.
$\qabsent$, rather than $\qlazy$, becomes the new identity of the $\coeffcolor{+}$ operator,
while the result of adding any two of the previously existing attributes remains the same
as in Table \ref{table:plus}.
We introduce the shorthand $\absentvec$ for the vector mapping all variables to $\qabsent$.

We do, however, need to introduce a new operation on $\coeffsymbol$s and $\coeffcolor{\gamma}$s.
When suspending a computation in the $\rref{T-Thunk}$ rule of the
original \cbpvlangname, it was sufficient
to simply use the lazy vector $\lazyvec$ for the required attributes,
as no variables are used strictly by a thunked computation.
If we wish to be more precise, however, it is inaccurate to say that a
thunk that does not use a variable $x$ in its body is lazy in $x$;
we must instead check the thunk construct with a vector that both is lazy
with respect to the thunk body and does not introduce any new usages.

To achieve this, we introduce a new operation $\lazify{\cdot}$ (read ``lazify'')
that makes a usage lazy. $\lazify{\coeffsymbol} = \qlazy$ for all $\coeffsymbol$ except $\qabsent$; instead,
$\lazify{\qabsent} = \qabsent$ because suspending an unused variable does not introduce a usage of that variable.
This new operation lifts pointwise to attribute vectors.
We use this new operation to define new rules for thunks,
presented \ifextended along with a handful of other minor rule changes \fi in Figure \ref{fig:cbpv-absent}.
All the other rules remain unchanged\ifextended\else, save for replacing $\lazyvec$ with $\absentvec$ in the rules for variables and $()$\fi.

This extension to \cbpvlangname enjoys all the
metatheoretic properties proven previously,
without significant changes to any proofs or definitions.
We therefore elide these proofs for brevity,
noting only that Theorem \ref{thm:soundness-bot} also applies
to variables with a $\qabsent$ attribute in addition to those with an $\qlazy$.

The simplicity of this extension
is a compelling demonstration of the benefits of reasoning about variable usage and strictness
in CBPV, rather than in CBN directly.

\begin{figure}
  \ifextended
  $$
  \infer[\rlabel*{T-Var-Ext}]{\absentvec \cdot \Gamma_1, x :^{\qstrict} A, \absentvec \cdot \Gamma_2 \vdash x : A}{} \qquad
  \infer[\rlabel*{T-Unit-Ext}]{\absentvec \cdot \Gamma\vdash () : \texttt{unit}}{}
  \vspace{2.75ex}
  $$
  \fi
  $$
  \infer[\rlabel*{T-Thunk-Ext}]{\lazify{\gamma} \cdot \Gamma \vdash \{ M \} : \utype{\gamma} B}
  {\coeffcolor{\gamma} \cdot \Gamma \vdash M : B}
  \ifextended \vspace{2.75ex} $$ \else \qquad \fi
  \ifextended
  $$
  \infer[\rlabel*{E-Var-Ext}]
  {\absentvec \cdot \rho_1, x :^{\qstrict} W, \absentvec \cdot \rho_2 \vdash x \Downarrow W}
  {}  \qquad
  \infer[\rlabel*{E-Unit-Ext}]{\absentvec \cdot \Gamma\vdash () \Downarrow ()}{}
  \vspace{2.75ex}
  $$
  $$
  \fi
  \infer[\rlabel*{E-Thunk-Ext}]
  {\lazify{\gamma} \cdot \rho \vdash \{ M \} \Downarrow \{ \coeffcolor{\gamma}, \rho , M \}}
  {}
  $$
  \ifextended \vspace*{-0.5\baselineskip} \else \vspace*{-1\baselineskip} \fi
  \ifextended \caption{Selected static and semantic rules for \cbpvlangname with unused variable tracking} \else
  \caption{Static and semantic rules for thunks in \cbpvlangname with unused variable tracking} \fi
  \ifextended \vspace*{-1\baselineskip} \else \vspace*{-1\baselineskip} \fi
  \label{fig:cbpv-absent}
\end{figure}

\subsection{Extending \cbnlangname}

\newcommand{\wf}{\textsc{wf }}

\newcommand{\cbnwfprop}{
  \begin{multicols}{2}
    $$
    \infer
    {\coeffcolor{\gamma} \vdashwf \tau}
    {\wf \tau & \coeffcolor{\lazify{\gamma} = \lazify{\gamma + \coeffsof{\tau}}}}
    \vspace{2.75ex}
    $$
    $$
    \infer
    {\wf \Gamma}
    {\forall x : \tau^{\coeffcolor{\gamma}} \in \Gamma, \coeffcolor{\gamma} \vdashwf \tau}
    $$
    \begin{align*}
      \coeffsof{\texttt{unit}} &= \absentvec \\
      \coeffsof{\tau_1^{\gamma_1} \times \tau_2^{\gamma_2}} &=
      \coeffcolor{\gamma_1 + \gamma_2 + \coeffsof{\tau_1} + \coeffsof{\tau_2}} \\
      \coeffsof{\tau_1^{\gamma_1} + \tau_2^{\gamma_2}} &=
      \coeffcolor{\gamma_1 + \gamma_2 + \coeffsof{\tau_1} + \coeffsof{\tau_2}} \\
      \coeffsof{(x:^{\coeffsymbol} \tau_1^{\gamma_1}) \xrightarrow{\gamma_2} \tau_2} &=
      \coeffcolor{\gamma_2 + \coeffsof{\tau_1} + \downshift{x} \coeffsof{\tau_2}}
    \end{align*}
  \end{multicols}
  \vspace*{-\baselineskip}
  $$
  \infer
  {\wf \texttt{unit}}
  {} \qquad
  \infer
  {\wf \tau_1^{\coeffcolor{\gamma_1}} \times\tau_2^{\coeffcolor{\gamma_2}}}
  {\coeffcolor{\gamma_1} \vdashwf \tau_1 & \coeffcolor{\gamma_2} \vdashwf \tau_2} \qquad
  \infer
  {\wf \tau_1^{\coeffcolor{\gamma_1}} +\tau_2^{\coeffcolor{\gamma_2}}}
  {\coeffcolor{\gamma_1} \vdashwf \tau_1 & \coeffcolor{\gamma_2} \vdashwf \tau_2} \qquad
  \infer
  {\wf (x:^{\coeffsymbol} \tau_1^{\coeffcolor{\gamma_1}}) \xrightarrow{\coeffcolor{\gamma_2}} \tau_2}
  {\coeffcolor{\gamma_1} \vdashwf \tau_1 & \coeffcolor{\gamma_2} \vdashwf \downshift{x} \tau_2}
  $$
}
Tracking unused variables in \cbnlangname requires significantly more effort
due to implicit thunking.
\ifextended
\begin{figure}
  \cbnwfprop
  \caption{Well-formedness requirements for \cbnlangname types}
  \vspace*{-0.5\baselineskip}
  \label{fig:cbn-wf}
\end{figure}
\fi
In particular, types in the extended \cbnlangname must
satisfy a certain well-formedness condition:
a type cannot claim that it uses more variables than those
used by the derivation that produces it.

Consider the term $(x, y)$.
Assuming $x$ and $y$ both have the type $\texttt{unit}$, we can
assign this term type
$\texttt{unit}^{\coeffcolor{x:\qstrict, y:\qabsent}} \times
\texttt{unit}^{\coeffcolor{x:\qabsent, y:\qstrict}}$,
producing effect $\coeffcolor{x : \qlazy, y : \qlazy}$.
This makes intuitive sense:
despite the fact that each side of the pair only uses one of the two
variables, the pair as a whole uses both $x$ and $y$.
It would be impossible, on the other hand,
to produce this type while not using either of the two variables;
to create either side of this pair type, we would need an expression
that uses one of the two strictly.
In general, typing derivations that look like
$\Gamma \vdash e :^{\coeffcolor{x:\qabsent, \dots}}
\tau_1^{\coeffcolor{x:\qlazy, \dots}} \times\tau_2^{\coeffcolor{\gamma}}$
are not possible: there is no way to produce a type that uses $x$
without using $x$ in its derivation.

\newcommand{\cbnexttyping}{
  $$
  \infer[\rlabel*{T-CBN-Var-Ext}]
  {\Gamma \cbnvdash x :^{\coeffcolor{\gamma, x:\qstrict}} \tau}
  {x : \tau^{\coeffcolor{\gamma}} \in \Gamma} \qquad
  \infer[\rlabel*{T-CBN-Abs-Ext}]
  {\Gamma \cbnvdash \lambda x. e :^{\lazify{\gamma_1 + \gamma_2}}
  (x :^{\coeffsymbol} \tau_1^{\coeffcolor{\gamma_1}}) \xrightarrow{\coeffcolor{\gamma_2}} \tau_2}
  {\deduce{\Gamma, x :^{\coeffcolor{\gamma_1}} \tau_1 \cbnvdash e :^{\coeffcolor{\gamma_2, x : \coeffsymbol}} \tau_2}{
      \deduce{\coeffcolor{\gamma_1} \vdashwf \tau_1}{
        \deduce{}{
        \coeffcolor{\lazify{\downshift{x} \coeffsof{\tau_2}} = \lazify{\gamma_1 + \downshift{x} \coeffsof{\tau_2}} }}
  }}}
  \vspace{2.75ex}
  $$
  $$
  \infer[\rlabel*{T-CBN-Inl-Ext}]
  {\Gamma \cbnvdash \textbf{inl } e :^{\lazify{\gamma_1 + \gamma_2 + \coeffsof{\tau_2}}} \tau_1^{\coeffcolor{\gamma_1}} + \tau_2^{\coeffcolor{\gamma_2}}}
  {\Gamma \cbnvdash e :^{\coeffcolor{\gamma_1}} \tau_1 &
    \coeffcolor{\lazify{\gamma_1} = \lazify{\gamma_2}} &
    \coeffcolor{\gamma_2} \vdashwf \tau_2
  }
  \vspace{2.75ex}
  $$
  $$
  \infer[\rlabel*{T-CBN-Inr-Ext}]
  {\Gamma \cbnvdash \textbf{inr } e :^{\lazify{\gamma_1 + \gamma_2 + \coeffsof{\tau_1}}} \tau_1^{\coeffcolor{\gamma_1}} + \tau_2^{\coeffcolor{\gamma_2}}}
  {\Gamma \cbnvdash e :^{\coeffcolor{\gamma_2}} \tau_2 &
    \coeffcolor{\lazify{\gamma_1} = \lazify{\gamma_2}} &
    \coeffcolor{\gamma_1} \vdashwf \tau_1
  }
  \vspace{2.75ex}
  $$
  $$
  \infer[\rlabel*{T-CBN-App-Ext}]
  {\Gamma \cbnvdash e_1\;e_2 :^{\coeffcolor{\gamma_1 + \gamma_3 + \lazify{\gamma_2}}}  \downshift{x} \tau_2}
  {\Gamma \cbnvdash e_1 :^{\coeffcolor{\gamma_1}}
    (x :^{\coeffsymbol} \tau_1^{\coeffcolor{\gamma_2}}) \xrightarrow{\coeffcolor{\gamma_3}} \tau_2  &
    \Gamma \cbnvdash e_2 :^{\coeffcolor{\gamma_2}} \tau_1
  }
  \vspace{2.75ex}
  $$
  $$
  \infer[\rlabel*{T-CBN-Let-Ext}]
  {\Gamma \cbnvdash \textbf{let } x = e_1 \textbf{ in } e_2 :^{\coeffcolor{\lazify{\gamma_1} + (\downshift{x} \gamma_2)}} \downshift{x} \tau_2}
  {\Gamma \cbnvdash e_1 :^{\coeffcolor{\gamma_1}} \tau_1 &
  \Gamma, x :^{\coeffcolor{\gamma_1}} \tau_1 \cbnvdash e_2 :^{\coeffcolor{\gamma_2}} \tau_2}
  $$
}

\ifextended
We can express this property formally by defining a function on types $\coeffsof{\cdot}$ (read "effects of"),
that collects all the latent effects on a type that appear in positive positions
(i.e., it does not include the latent effects on function arguments), and adds them together.
We then have the following definition of a \textit{valid typing derivation}:

\begin{definition}[Valid Typing Derivation] \label{def:valid}
  A \cbnlangname typing derivation $\Gamma \cbnvdash e :^{\coeffcolor{\gamma}} \tau$
  is \textit{valid} if $\coeffcolor{\gamma} \vdashwf \tau$.
\end{definition}

This makes formal the intuitive requirement we outlined earlier: a type cannot
declare that it uses more variables than the ones produced by derivation used to check it.
We can use Definition \ref{def:valid} to understand the definition of $\vdashwf$ in Figure \ref{fig:cbn-wf}:
well-formed types are those that can be produced by valid derivations.

\begin{figure}
  \cbnexttyping
  \caption{Main typing rules for extended \cbnlangname}
  \vspace*{-0.5\baselineskip}
  \label{fig:cbn-typing-extended}
\end{figure}
\fi

Types produced by the \cbnlangname typing rules should have this property,
but unfortunately there are a handful of rules where
types are chosen without an accompanying subderivation,
such as the rules for $\textbf{inl}$ or  $\textbf{inr}$.
In these cases, the rules must enforce that any such types are well formed and
satisfy particular usage conditions to
ensure that all derivations
permitted by the \cbnlangname typing judgment \ifextended are valid\else contain well-formed types\fi.
\ifextended

The extended rules, a selection of which can be found in Figure \ref{fig:cbn-typing-extended},
allow us to prove the following lemma,
which is necessary for the correctness of the extended \cbnlangname translation:

\begin{lemma}[Well-Formed Derivations are Valid]
  \label{lem:valid}
  If $\Gamma \cbnvdash e :^{\coeffcolor{\gamma}} \tau$ and $\wf \Gamma$,
  then $\coeffcolor{\gamma} \vdashwf \tau$.
\end{lemma}
\begin{proof}
  Via induction on the typing derivation.
\end{proof}
Beyond the well-formedness requirements for types,
the other major change to the typing rules in the extended \cbnlangname
is how the $\rref{T-CBN-Let-Ext}$ and $\rref{T-CBN-App-Ext}$ rules handle
the effects associated with the bound (or argument) expression.
The old versions of these rules just put these effects into the typing context, associated
with the variable to which their producing expression was being bound.
If we wish to distinguish between lazy use and non-use, however, these rules must
also account for variables that are used in bound expressions that are never evaluated.
For example, in the term $\textbf{let } x = y \textbf{ in } ()$,
the variable $y$ is used lazily.
Just adding the effect produced by the $y$ subexpression
to the context used to check $()$ would not capture this fact,
and so instead the extended rule
must also lazify the effect of checking $y$ and add that to the
overall effect of the let-binding. Function applications handle their arguments similarly.
\else
The formal definition of these well-formedness requirements,
along with a selection of the extended rules for \cbnlangname,
can be found in the extended version of this paper  \cite{extended-version}.
\fi

The translation between \cbnlangname and \cbpvlangname
needs only a small adjustment:
the type translation presented in Figure \ref{fig:cbn-trans} must lazify all the
intermediate effects on each type rather than returning $\lazyvec$ for any non-function types.
With this modification, the translation enjoys the same
correctness properties described in Section \ref{subsec:translation} for all
well-formed \cbnlangname derivations.

\section{Related Work} \label{sec:related}

\emph{Strictness analysis}
was first presented by \citet{mycroft-theory-1980}.
His initial work laid the foundations for strictness analysis using abstract interpretation \cite{abstract-interpretation},
but it was limited to programs working with simple data
for which it sufficed to use an abstract domain containing just a $\top$ (defined) and a $\bot$ (undefined) element.
The original formulation was extended
to a four-point domain to better handle list programs \cite{wadler-strictness-1987}.
In addition to fully defined and undefined values,
the four-point domain enabled reasoning about data that was itself defined but that
might contain subcomponents that were not
(e.g., a fully-defined list whose elements were undefined).

\citet{wadler-projections-1987} improved upon earlier
techniques using \textit{projections} \cite{hughes-strict} to handle significantly
more complex programs and bypass the need for abstract interpretation.
Projection-based analysis involves reasoning backwards to determine how defined a program's inputs
need to be, based on how its outputs are used;
it was first implemented in practice by \citet{kubiak-implementing-1992}
and it now forms the foundation for strictness analysis in GHC \cite{sergey-theory-2014}.
Other approaches have been found that outperform GHC, such as the ``Optimistic Evaluation''
strategy \cite{Ennals-Jones-2003}, but these have proven too complex for practical adoption.

\textit{Demand analysis} \cite{launchbury-demand}
generalizes strictness analysis
by allowing compilers to reason about ``how much'' of a value is used (or ``demanded'').
It has been adapted to provide fine-grained analysis
of performance in non-strictly evaluated languages \cite{xia-story-2024, bjerner-demand}
and to guide compiler optimizations \cite{sergey-modular-2014}.
\cbnlangname's types also describe demand beyond the top-level constructor;
the suspended $\coeffcolor{\gamma}$s that appear in types
describe how variables get used as more of the result of a term is demanded.

%Other tools have been used to surface strictness to programmers.
%Property-based testing \cite{foner-keep-2018} can be used
%to specify and test the strictness of functions,
%while other tools exist to visualize laziness \cite{devries-visualizing}
%and encourage writing code both more lazily \cite{Chitil-2011} and more strictly \cite{nothunks}.

\subsubsection*{Strictness in Types}

Usage type systems \cite{turner-1995,wansbrough-1999, wansbrough-2005} such as linear types \cite{girard-linear, maraist-linear, linear-types-wadler,linearity-haskell}
are used to track how variables and values are used in programs.
Attempting to apply such systems to modeling strictness, however, breaks down almost immediately
when considering variables that appear in the bodies of functions.
For example, such systems would just describe \texttt{f4} from the examples
in Figure \ref{fig:strict-example-1} as using its argument $y$,
failing to capture the fact that \texttt{y} is not strictly used.

Approaches based on information flow
\cite{denning-denning-info,volpano-secure-flow,Zdancewic-2001,zdancewic-secure-002,Palsberg-flow}
have the opposite problem.
Unlike usage typing, which underapproximates strictness,
information-flow typing overapproximates it.
While these systems can distinguish the ways in which \texttt{f1} and \texttt{f4} in Figure \ref{fig:strict-example-1} use their arguments,
they also distinguish between variable usages that \textit{should} be considered equivalent
from a strictness perspective.
For example, an information-flow system would say that
\textbf{if} \texttt{y} \textbf{then} \texttt{1} \textbf{else} \texttt{2}
leaks information about \texttt{y}
and that \texttt{y} \`{}\textbf{seq}\`{} \texttt{1} does not,
since the latter produces the same result regardless of \texttt{y}'s value.
It is clear, however, that both of these programs use \texttt{y} strictly,
so information flow is not sufficient for modeling strictness either.

There are, however, three existing source-level approaches to tracking strictness.
Kuo and Mishra \cite{kuo-87,kuo-89} describe a constraint-gathering static analysis
they call ``strictness types.''
%Despite its name, however, their analysis is separate from types as we understand them;
%it is designed to be used alongside a type system, rather than
%integrated with one. \bcp{Not sure I got that.}
%
\citet{schrijvers-strictness-2010} describe
an effect system for modeling strictness
and use it to guide a handful of program optimizations.
Both of these approaches, however, are limited to
analysis of flat data (i.e., numbers or other base types) and thus
are not sufficiently expressive for our purposes.

\citet{barendsen-strictness-2007} and \citet{smetsers-strictness} outline
a type system modeling strictness using attributes annotated on types.
Their system features a two-point lattice containing
a $!$ attribute for strict usage and a $?$ attribute for uses
lacking strictness information.
This approach allows
function types to be annotated with usage information
(e.g., a strict function from $A$ to $B$ would be typed $A^!\rightarrow B$),
but it only tells half the story.
Guarantees can only be made about strict
usage (i.e., types for which a $!$ attribute is inferred) whereas
\cbnlangname and \cbpvlangname can make guarantees about the evaluation of lazy functions as well.

\subsubsection*{Effects and Coeffects}

Type-and-effect systems are well-studied tools for modeling side effects at the level of type systems
\cite{lucassen-polymorphic-1988,katsumata-2014,talpin-type-1992, Wadler-effect-2003}.
The study of coeffects is relatively newer,
having been first introduced by \citet{petricek-coeffects-2013} in 2013,
but it has been an area of active study for the past decade
\cite{Brunel-coeffect-2014, Petricek-coeffect-2014, Orchard-graded-2019, choudhury-graded-2021,ghica-linear}.
Coeffects have been used to track properties like differential privacy \cite{reed-privacy},
error bounds \cite{kellison-2025},
irrelevance in dependent type theories \cite{abel-coeffects}, and
resource usage \cite{torczon-effects-2024, choudhury-graded-2021,Bianchini-resource-2023,linearity-haskell}.

The design of \cbnlangname and \cbpvlangname resembles coeffect-graded systems,
such as those described by \citet{Bianchini-resource-2023} and \citet{choudhury-graded-2021}.
However, while grade algebras can describe properties similar to laziness and strictness,
it is not clear whether coeffect-graded systems generalize our work.
In particular, \cbnlangname types describe how variables will be used if the values inhabiting those types are used.
In a graded system, on the other hand, grades describe only how values are used, not the downstream effects of such usage.
\cbnlangname is more detailed because it annotates types with an entire grade vector
instead of a single grade annotation on a graded modal type,
but it is possible that a clever choice of grade algebra may be able to capture the same properties.
Further work may provide a more rigorous understanding of the relationship between these systems.

\ifpoly
The original treatment of polymorphic effects by \citet{lucassen-polymorphic-1988}
was simplified by \citet{rytz-lightweight-2012},
and this simpler version informed our presentation of polymorphic coeffects.
\fi

The interaction between effects and coeffects is an active research area
\cite{dal-lago-2022, gaboardi-combining-2016}.
%\citet{dal-lago-2022} develop the notion of
%\textit{corelators} to provide an axiomatic semantics for reasoning about the equivalence of
%programs in languages with effects and coeffects.
\citet{nanevski-dynamic-2003} presents a type system that uses modal types
to model local and global state in a manner analogous to effects and coeffects.
Of particular relevance to our work is that the $\Box$ modality behaves like the $\utype{}$ type
in \cbpvlangname, in the sense that it both suspends computation
and tracks local variable usage within the type, which has the flavor of a coeffect.
\citet{gaboardi-combining-2016} use graded modal types to describe a type system in
which effects and coeffects can interact with one another according to distributive laws,
while \citet{Hirsch-2018} investigate how effects and coeffects interact
when distributive laws between them do not exist. The latter work is especially relevant because their
choice of how to layer effects and coeffects semantically results in
either strict or lazy evaluation, although they do not surface this information statically.
Lastly, \citet{McDermottMycroft-2018} blend coeffects with effects to implement
an effect system for a call-by-need language similarly to \cbnlangname.

\subsubsection*{Call-By-Push-Value}

Call-by-push-value and its translations from call by name and call by value are due to Levy
\cite{levy-call-by-push-value-1999,levy-call-by-push-value-2001,levy-call-by-push-value-2006,levy-cbpv-2003,Levy-cbpv-2022}.
It is frequently used as a substrate for studying effects \cite{mcdermott-grading-2025,rioux-computation-2020,kammar-2013,kammar-2012} and their interplay with other
systems, including evaluation order \cite{mcdermott2019extended}, dependent types \cite{pedrot-fire-2019},
and time complexity \cite{kavvos-cbpv}.
Torczon et al. \cite{torczon-effects-2024} extend CBPV with support
for both effects and coeffects;
their presentation and mechanization of this system heavily influenced this paper.
They, in turn, drew from \citet{forster-call-by-push-value-2019}'s mechanization of CBPV in Rocq.

\section{Conclusion and Future Work} \label{sec:future}

We have presented a new intensional definition of strictness
that refines the usual extensional definition by describing usage more precisely.
We introduced \cbnlangname and \cbpvlangname,
type systems that embody our definition using effects,
related these systems via a type-preserving translation,
and proved that they capture appropriate semantic notions of strict
and lazy usage.

In the future, we hope to investigate whether \cbnlangname
could be used in practice to typecheck real programs in languages like Haskell.
In its current presentation, the annotation burden in the \cbnlangname type system is quite severe---expecting programmers to annotate types with the usage of every variable
in scope does not seem reasonable!
%However, while the $\coeffcolor{\gamma}$s on types
%are necessary for \cbnlangname to compute attributes for all variables,
%it is likely sufficient to surface only the resulting attributes to programmers
%for them to understand the strictness of their programs,
%hiding the intermediate $\coeffcolor{\gamma}$s.
Prior work by \citet{wansbrough-2005}, however, describes implementing
a type system with a similar annotation burden
practically in GHC; we plan to explore whether the
principles guiding that work can be applied to \cbnlangname.

To work for realistic Haskell programs,
a few refinements to \cbnlangname's type system would be needed.
\ifpoly In particular, \else
As mentioned in Section \ref{sec:cbn}, the requirement that \cbnlangname function types
declare the effects that they allow their arguments to produce is quite restrictive.
Existing work on systems similar to \cbnlangname \cite{McDermottMycroft-2018} addresses
this limitation using effect polymorphism
\cite{rytz-lightweight-2012,lucassen-polymorphic-1988}, and we
believe the same approach would work here.
This would involve building out the theory and metatheory
of effect-polymorphic extensions of \cbnlangname and \cbpvlangname.
\citet{rioux-computation-2020} describe $\text{CBPV}^{\forall}\!$,
a type-polymorphic extension to CBPV, providing a promising foundation for
this endeavor.

Further,
\fi
our type system would need to support more complex datatypes, such as lists.
Indeed, we believe the main principles should extend relatively
straightforwardly.
Aside from the usual \textbf{nil} and \textbf{cons} constructors for
introducing lists,
we would add a $\textbf{fold}\;f\;lst\;acc$ operation to consume them.
The strictness annotations on the type inferred for $f$ would describe
how the list and accumulator are used: if $f$ has an $\qstrict$ attribute for both arguments,
then the list contents are all strictly used (e.g., the \texttt{sum} example from Figure \ref{fig:strict-example-fold}),
whereas an $\qlazy$ attribute for the list element argument and
an $\qstrict$ attribute for the accumulator
would result in only the ``spine'' of the list being strictly used (e.g., the \texttt{length} function).
An $\qlazy$ attribute for the accumulator would mean that the fold does not even
use the spine of the list strictly (e.g., \texttt{map}).

We would also like to explore whether \cbpvlangname
can be used to guide compiler optimizations.
CBPV is often used as a reasoning tool for intermediate languages \cite{downen-2020,mercer-2022, new-cbpv},
so the groundwork has already been laid.
We plan to build a compiler for a fragment of Haskell into \cbpvlangname to investigate
what transformations can be applied based on intensional strictness in this setting.

Additionally, as mentioned previously, the call-by-name insights provided by this work
extend to call-by-need evaluation,
as the two strategies use variables with the same strictness.
However, there is interesting future work in actually formalizing this relationship,
and in particular in working out the details of how to model call-by-need
evaluation in CBPV.
One option would be to use the Extended Call-By-Push-Value system of \citet{mcdermott2019extended},
which changes the evaluation strategy of CBPV itself to be call-by-need. Another option
would be to encode call-by-need evaluation via translation to CBPV by extending the
latter with mutable references. In such a translation, we conjecture that all value types $A$
would be translated to references to sum types $\textbf{ref} (\utype{}\ftype A + A)$,
and evaluating thunks the first time would translate to storing their resulting value
in the reference in place of the previously thunked expression.

Beyond strictness, the characterization of $\coeffcolor{\gamma}$s in
\cbnlangname and \cbpvlangname
deserves further attention: we would like to study more formally the connection between
\cbpvlangname's $\coeffcolor{\gamma}$s and coeffects.
We have also noted the potential relationship between these types and the logic of modal necessity,
and we are eager to learn more about the extent and nature of this
relationship.

\section*{Acknowledgements}
We thank the anonymous reviewers for their feedback and suggestions.
We also thank José Manuel Calderón Trilla and Simon Peyton-Jones
for their valuable suggestions about how to best frame and explain the contributions of
the paper.
Lastly, we thank Jonathan Chan, Yiyun Liu, Nick Rioux, Cassia Torczon, and PLClub at large
for feedback on the paper and help with formalization.

This material is based upon work supported by the U.S. National Science Foundation
under Grants No. 2313998 and 2402449. Any opinions, findings, and conclusions or recommendations
expressed in this material are those of the authors and do not
necessarily reflect the views of the U.S. National Science Foundation.

\bibliographystyle{ACM-Reference-Format}
\bibliography{main}

\end{document}